\documentclass[twoside]{article}

\usepackage[preprint]{aistats2026}
\usepackage[round]{natbib}

\usepackage[utf8]{inputenc} 
\usepackage[T1]{fontenc}    
\usepackage{hyperref}       
\hypersetup{
    colorlinks,
    linkcolor={red!50!black},
    citecolor={blue!50!black},
    urlcolor={blue!80!black}
}
\usepackage{url}            
\usepackage{booktabs}       
\usepackage{amsfonts}       
\usepackage{nicefrac}       
\usepackage{microtype}      
\usepackage{lipsum}
\usepackage{graphicx}
\usepackage{subcaption}
\usepackage{enumitem}

\usepackage{xcolor}

\usepackage{amsmath}
\usepackage{amssymb}

\usepackage{algorithm}
\usepackage{algpseudocodex}

\usepackage{amsthm}
\usepackage{mathtools}  

\usepackage{pgfplots}
\pgfplotsset{compat=newest}
\usepackage{tikz}
\usepackage{wrapfig}

\usepackage{acronym}
\acrodef{OLS}{ordinary least squares estimator}
\acrodef{MLE}{maximum likelihood estimator}
\acrodef{SNR}{signal-to-noise ratio}
\acrodef{LTI}{linear time-invariant}
\acrodef{LQR}{linear quadratic regulator}
\acrodef{CE}{certainty equivalence}
\acrodef{PE}{persistency of excitation}

\newtheorem{theorem}{Theorem}[section]
\newtheorem{corollary}[theorem]{Corollary}

\newtheorem{lemma}[theorem]{Lemma}
\newtheorem{remark}{Remark}[section]
\newtheorem{example}{Example}[section]
\newtheorem{manualexampleinner}{Example}
\newenvironment{manualexample}[1]{%
  \IfBlankTF{#1}
    {\renewcommand{\themanualexampleinner}{\unskip}}
    {\renewcommand\themanualexampleinner{#1}}%
  \manualexampleinner
}{\endmanualexampleinner}

\newtheorem{assumption}{Assumption}
\newtheorem{prop}[theorem]{Proposition}
\newtheorem{definition}{Definition}

\DeclareMathOperator*{\argmax}{arg\,max}

\newcommand{\R}{\mathbb{R}}

\newcommand{\Prob}{\mathbb{P}}
\newcommand{\Expect}{\mathbb{E}}
\newcommand{\N}{\mathcal{N}}
\newcommand{\SSet}{\mathcal{S}}

\newcommand{\E}{\mathcal{E}}
\newcommand{\subG}[1]{\mathrm{subG}\left({#1}\right)}
\newcommand{\subExp}[1]{\mathrm{subExp}\left({#1}\right)}

\newcommand{\simiid}{\stackrel{\text{i.i.d.}}{\sim}}
\newcommand{\KL}[2]{\mathrm{KL}({#1} \Vert {#2})}
\newcommand{\kl}[2]{\mathrm{kl}({#1} \Vert {#2})}
\newcommand{\diag}[1]{\mathrm{diag}_{#1}}
\newcommand{\tr}[1]{\mathrm{tr}\left({#1}\right)}

\begin{document}

%

%

\twocolumn[

\aistatstitle{High Effort, Low Gain: Fundamental Limits of Active Learning for Linear Dynamical Systems}

\aistatsauthor{ Nicolas Chatzikiriakos\textsuperscript{*} \And Kevin Jamieson\textsuperscript{$\dagger$}  \And  Andrea Iannelli\textsuperscript{*}}
\runningauthor{Nicolas Chatzikiriakos, Kevin Jamieson,  Andrea Iannelli}

\aistatsaddress{University of Stuttgart \And University of Washington \And University of Stuttgart}
]

{
\renewcommand{\thefootnote}{\fnsymbol{footnote}}
\footnotetext[1]{Institute for Systems Theory and Automatic Control, \url{nchatzi@ist.uni-stuttgart.de,} \url{iannelli@ist.uni-stuttgart.de}}
\footnotetext[2]{Allen School of Computer Science \& Engineering, \url{jamieson@cs.washington.edu}}
}

\begin{abstract}
  We consider the problem of identifying an unknown linear dynamical system from a finite hypothesis class. In particular, we analyze the effect of the excitation input on the sample complexity of identifying the true system with high probability. 
  To this end, we present sample complexity lower bounds that capture the choice of the selected excitation input.
  The sample complexity lower bound gives rise to a system-theoretic condition to determine the potential benefit of experiment design.        
  Informed by the analysis of the sample complexity lower bound, we propose a persistency of excitation (\acs{PE}) condition tailored to the considered setting, which we then use to establish sample complexity upper bounds. Notably, the \acs{PE} condition is weaker than in the case of an infinite hypothesis class and allows analyzing different excitation inputs modularly. 
  Crucially, the lower and upper bounds share the same dependency on key problem parameters.
  Finally, we leverage these insights to propose an active learning algorithm that sequentially excites the system optimally with respect to the current estimate, and provide sample complexity guarantees for the presented algorithm.
  Concluding simulations showcase the effectiveness of the proposed algorithm. 
\end{abstract}

\section{Introduction}
The problem of learning a model of an unknown dynamical system from data is important across domains such as reinforcement learning, data-driven control and robotics.
Hereby, it is of particular interest to obtain an accurate model with high confidence from as few samples as possible, since data collection is often expensive. 
To tackle this problem, there exists a large literature on experiment design both in dynamical systems and learning theory.
However, experiment design algorithms can be computationally demanding and
a study of the potential benefit of experiment design algorithms is lacking in the literature of dynamical systems.
This raises the question of whether experiment design algorithms should be applied universally, or whether a more nuanced answer is appropriate.
\\
Motivated by these challenges, we consider the setting where the learner has prior knowledge of the true system through a finite hypothesis class. This reflects cases where some prior knowledge
is available, e.g., based on first principles, yet certain parameters are hard to model or vary across
different instances. 
When it comes to dynamical systems, the data is usually collected from a single trajectory and hence is highly correlated. 
This poses a key challenge when analyzing learning in the finite sample regime. 
In particular, existing works considering the identification of dynamical systems\footnote{We provide an extended overview of the related literature in Appendix~\ref{appSection:relatedWorks}.} 
mostly consider the case of an infinite hypothesis class and rely on the least squares estimator. 
Early works derived sample complexity upper bounds for Gaussian inputs for linear~\citep{sarkar2019near, simchowitz2018learning} and certain classes of non-linear systems~\citep{foster20a,sattar2022Bilinear, sattar2022non}.
These results have been complemented by the sample complexity lower bounds for Gaussian excitations and linear systems presented by~\cite{jedra2022finite, tsiamis2021linear}.
For the case of a finite hypothesis class, \cite{chatzikiriakos2024b, muehlebach2025} provide sample complexity upper and lower bounds when the excitation is Gaussian. 

The problem of experiment design has a long history in system identification (see, e.g., \cite{bombois2011optimal,1977dynamic}). While these classical works consider the asymptotic case, recent works provided a finite sample perspective on the topic for linear~\citep{wagenmaker2020active, wagenmaker21a, chatzikiriakos2025convex} and certain classes of non-linear systems~\citep{lee2024active, mania2022active}. 
While the underlying principles in the finite-sample analysis are always similar, there exists no modular framework for the analysis. 
Further, the benefit of experiment design over randomly exciting the systems has not been explicitly investigated, even for linear dynamical systems. 
This work addresses these gaps by answering the following questions to improve our understanding of experiment design and active learning:
\vspace{-1em}
\begin{itemize}[align=parleft, labelsep=0.2em, itemindent=0.005em, itemsep=-0.1em]
    \item[Q1)]How large is the problem-specific benefit of oracle experiment design over isotropic Gaussian excitations in terms of the sample complexity? 
    \item[Q2)] How can sample complexity upper bounds be established for general excitation inputs?
    \item[Q3)] How can active learning algorithms generate excitation inputs to provably outperform random isotropic excitations?
    \item[Q4)] How does a finite hypothesis class simplify experiment design compared to an infinite hypothesis class?
\end{itemize}
\vspace{-1em}
With respect to these key questions, our primary contributions can be summarized as follows:
\vspace{-0.3em}
\begin{enumerate}[align=parleft, labelsep=0.03em, itemindent=0.155em, itemsep=-0.1em]
    \item We provide instance-dependent lower bounds on the sample complexity of identifying the true system with high probability. To this end, we leverage tools from information theory and derive the optimal oracle excitation input during data collection. We analyze the system-theoretic quantities dictating how the sample complexity lower bounds differ for different excitation inputs.
    \item Building on the notion of \ac{PE}, we propose a modular framework to derive sample complexity upper bounds for different excitation inputs. 
    While the probability of miss-specification decays exponentially for all excitations satisfying \ac{PE}, we establish that the decay rate depends on system theoretic quantities that match the lower bounds. 
    \item Using the notion of \ac{PE} we establish an interpretable condition pointing out when experiment design using certainty equivalence is more efficient than random excitations. 
    \item Notably, our analysis also uncovers cases where the benefit of experiment design is small even in an oracle setup, indicating that the usefulness of experiment design needs to be evaluated on a problem-specific level.
\end{enumerate}
\emph{Notation:}
The $n$-dimensional simplex is denoted by $\Delta_n$.
We denote the set of symmetric positive \mbox{(semi-)}definite matrices of dimension $n\times n$ by $\mathbb{S}^{n}_{++}(\mathbb{S}_+^n)$. 
Given a vector $v\in \R^n$ and a matrix $M \in \mathbb{S}_{++}^n$ we define $\Vert v \Vert_M \coloneqq \sqrt{v ^\top M v }$. 
Given a matrix $M\in \R^{n \times n}$ we denote its largest eigenvalue by $\lambda_\mathrm{max}(M)$ and the mean of all eigenvalues by $\lambda_\mathrm{mean}(M) = \frac{1}{n}\tr{M}$, where $\tr{M}$ indicates the trace of the matrix $M$.
Given a sequence $\{u(t)\}_{t=t_0}^{t_1}$ we denote the stacked collection as $U = \begin{bmatrix}
    u(t_0)^\top & \dots & u(t_1)^\top
\end{bmatrix}^\top$, where the boundaries of the interval will be clear from the context. 
We use $\diag{t}(M)$ to denote a block-diagonal matrix which repeats the matrix $M$ on its diagonal $t$-times, i.e., $\diag{t}(M) \coloneqq I_{t} \otimes M$, where $\otimes$ denotes the Kronecker product. 

\section{Preliminaries}\label{sec:ProblemSetup}
We consider the unknown discrete-time \ac{LTI} dynamical system
\begin{equation}
    x(t+1) = A_* x(t) + B_* u(t) + w(t),\label{eq:TrueSysEvo}
\end{equation}
where $x(t) \in \R^{n_x}$, $u(t)\in \R^{n_u}$ are the state and the input of the system at timestep $t$, and $w(t)$ is process noise. We assume $w(t) \stackrel{\text{i.i.d}}{\sim} \mathcal{N}(0, \Sigma_w)$ with a known covariance matrix $\Sigma_w\in \mathbb{S}^{n_x}_{++}$. However, all the presented results and the proposed algorithm can be adapted to the case where $\Sigma_w$ is unknown and only an upper bound $\sigma_w^2 I_{n_x} \succeq \Sigma_w$ or an informative prior is available. 
For notational simplicity, we assume $x(0) = 0$ unless stated otherwise.
Further, we assume the learner has access to a finite set of systems, containing the true system, i.e.,
\vspace{-0.5em}
\begin{equation} \label{eq:defSet}
    \theta_* = (A_*, B_*) \in \SSet := \{(A_0, B_0), \dots, (A_N, B_N)\},    
\end{equation}
where $\SSet$ is known to the learner. 
We assume $(A_*, B_*) = (A_0, B_0)$ for notational simplicity.
The goal of the learner is to identify the true system matrices $\theta_*\in \SSet$ from a single data trajectory $\mathcal{D}_T = (\{x(t)\}_{t=0}^{T}, \{u(t)\}_{t=0}^{T-1})$, where the excitation input $u$ can be selected by the learner. 
Thus, in the context of this work, an identification problem is defined by the tuple $(\theta_*, \SSet)$.
The problem of identifying the true system from a finite hypothesis class naturally emerges in many settings, such as
fault detection~\citep{miljkovic2011fault} or validation of ecological and evolutionary models~\citep{Johnson2004}. 
More generally, the set $\SSet$ can be viewed as a principled, model-driven way to introduce knowledge of the real process into the identification problem.  

The performance of the learner is evaluated in terms of the samples required to identify the true system with confidence $\delta$. 
For our analysis, it is important to understand how predictions using  $\theta_*$ compare to predictions using $\theta_i$, for some $i \in [1,N]$. 
To this end, we denote the sum of one-step prediction errors between $\theta_*$ and $\theta_i$ by 
\begin{equation}\label{eq:predErrorTwoSys}
    \tilde{\varepsilon}_{\theta_i}(0, \tau) \coloneqq \sum_{t = 0}^{\tau-1} \Vert \Delta A_i x(t) + \Delta B_i u(t)\Vert_{\Sigma_w^{-1}}^2,
\end{equation}
where we defined $\Delta A_i \coloneqq A_* - A_i$ and $\Delta B_i \coloneqq B_* - B_i$. 
Recall that the data $\mathcal{D}_\tau$ is collected along a trajectory of the dynamical system~\eqref{eq:TrueSysEvo}.
Recursively plugging in the dynamics~\eqref{eq:TrueSysEvo} for $x(t)$ yields 
\begin{align}
    \tilde{\varepsilon}_{\theta_i}(0, \tau) =\sum_{t = 0}^{\tau-1} \bigg\Vert \Delta A_i \Big(
        &\sum_{s=0}^{t-1} A_*^{t-1-s} B_* u(s) \label{eq:epsilonReformulated} \\& + A_*^{t-1-s} w(s)\Big)  + \Delta B_i u(t) \bigg\Vert _{\Sigma_w^{-1}}^2\notag
\end{align}
which reveals that $\tilde{\varepsilon}_{\theta_i}(0, \tau)$ is a random quantity even if the excitation input is  deterministic\footnote{A detailed derivation, for the general case $x(0) \in \R^{n_x}$ that includes the case where $u(t)$ consists of both a deterministic and a random part is presented in Appendix~\ref{secAppendix:AdditionalResults}.}.
This is due to the process noise affecting the evolution of the system~\eqref{eq:TrueSysEvo}.
Taking the expectation and using $w(t) \simiid\N(0, \Sigma_w)$ it can be shown that 
\begin{small}
    \begin{align}
         \Expect\left[\tilde{\varepsilon}_{\theta_i}(0, \tau)\right] = &\Expect\Big[\sum_{t=0}^{\tau-1} \big\Vert \Delta A_i \sum_{s=1}^{t} A_*^{t-s} B_* u(s) + \Delta B_i u(t)  \big\Vert _{\Sigma_w^{-1}}^2\Big] \notag \\ 
        &+ \Expect\Big[\sum_{t=0}^{\tau-1} \big\Vert \Delta A_i \sum_{s=1}^{t} A_*^{t-s} w(s)\big\Vert _{\Sigma_w^{-1}}^2\Big].  \label{eq:epsilonTildeReformulated}
    \end{align}        
\end{small}
To simplify notation, we define the Toeplitz matrices 
\begin{equation} \label{eq:defToeplitz}
    \begin{aligned}
    S_u(\tau) &\coloneqq \begin{bmatrix}
        B_* & 0 & \dots & 0 \\
        A_*B_* & B_* & 0 &  \vdots \\
        \vdots & &  \ddots  & 0 \\
        A^{\tau-1}_* B_* &A_*^{\tau-2} B_*  & \dots &  B_*
    \end{bmatrix}, 
    \\
    S_w(\tau) &\coloneqq \begin{bmatrix}
        \Sigma_w^{\nicefrac12} & 0 & \dots & 0 \\
        A_* \Sigma_w^{\nicefrac{1}{2}}& \Sigma_w^{\nicefrac{1}{2}} & 0 &  \vdots \\
        \vdots & &  \ddots  & 0 \\
        A^{\tau-1}_* \Sigma_w^{\nicefrac{1}{2}} &A_*^{\tau-2} \Sigma_w^{\nicefrac{1}{2}}  & \dots &  \Sigma_w^{\nicefrac{1}{2}}
    \end{bmatrix}.
\end{aligned}
\end{equation}
These matrices are tightly connected to the controllability of $(A_*, B_*)$, and are commonly used in systems theory to map from the $\tau$-step noise and input trajectories to the corresponding state trajectory. 
Using \eqref{eq:defToeplitz}, computing the expectation in~\eqref{eq:epsilonTildeReformulated} and writing the statement in matrix form simplifies to 
\begin{align} 
    \Expect\left[\tilde{\varepsilon}_{\theta_i}(0, \tau)\right]&=  \Expect\Big[U^\top \big(R_{\Sigma_w^{-1}}^i(\tau) + S_u(\tau)^\top Q_{\Sigma_w^{-1}}^i(\tau) S_u(\tau) \notag \\ &+  N_{\Sigma_w^{-1}}^i(\tau)S_u(\tau) + (N_{\Sigma_w^{-1}}^i(\tau)S_u(\tau))^\top \big)U\Big] \notag\\ 
    &+ \tr{S_w(\tau)^\top Q^i_{\Sigma_w^{-1}} S_w(\tau)}, \label{eq:expectedErrorFinal}
\end{align}
where $U\in\R^{n_u \tau}$ and we introduced the block-diagonal matrices 
\begin{subequations}\label{eq:defMatrices}
\begin{align}
    Q_{\Sigma_w^{-1}}^i(\tau) &\coloneqq \diag{\tau}(\Delta A_i ^\top \Sigma_w^{-1} \Delta A_i), \\
    R_{\Sigma_w^{-1}}^i(\tau) &\coloneqq \diag{\tau}(\Delta B_i^\top \Sigma_w^{-1} \Delta B_i) \\
    N_{\Sigma_w^{-1}}^i(\tau) &\coloneqq \diag{\tau}(\Delta B_i ^\top \Sigma_w^{-1} \Delta A_i)
\end{align}
\end{subequations}
Note that, leveraging the system-theoretic interpretation of the Toeplitz matrices \eqref{eq:defToeplitz}, $U^\top R^i_{\Sigma_w^{-1}}(\tau) U = \sum_{t=0}^{\tau}\Vert \Delta B_i u(t)\Vert_{\Sigma_w^{-1}}^2$ measures to what extent the difference between $B_*$ and $B_i$ can be seen in the weighted prediction error when applying the input sequence $U$. 
Similarly, $U^\top S_u(\tau)^\top Q_{\Sigma_w^{-1}}^i(\tau) S_u(\tau)U$ measures to what extent the difference between $A_*$ and $A_i$ can be seen in the weighted prediction error when applying $U$ and $\tr{S_w(\tau)^\top Q^i_{\Sigma_w^{-1}} S_w(\tau)}$ measures the influence of the process noise on this quantity. 
Accounting for the cross term, for any $i \in [1, N]$, the matrices
\begin{equation}\label{eq:defWi}
\begin{aligned}
     W_i(\tau) &\coloneqq R_{\Sigma_w^{-1}}^i(\tau) + S_u(\tau)^\top Q_{\Sigma_w^{-1}}^i(\tau) S_u(\tau) \\ &\quad +  N_{\Sigma_w^{-1}}^i(\tau)S_u(\tau)  + (N_{\Sigma_w^{-1}}^i(\tau)S_u(\tau))^\top,
\end{aligned}
\end{equation}
measure the difficulty of distinguishing between $\theta_*$ and $\theta_i$ and the sensitivity of \eqref{eq:predErrorTwoSys} to changes in the excitation input. 

\section{Can active learning algorithms enhance identification?}
To analyze the influence of the data collection scheme on the sample complexity, we first derive a sample complexity lower bound that holds for any \emph{reasonable} algorithm. 
This instance-specific lower bound quantifies the hardness of learning in our setup and provides a nuanced answer on whether active learning algorithms can significantly reduce the sample complexity.
Further, in Section~\ref{sec:UpperBound} we provide a modular framework for deriving complementary sample complexity upper bounds.
In this work, we make the following assumption on the excitation input, which is standard in the relevant literature (see, e.g., \cite{wagenmaker2020active})
\begin{assumption}[Bounded input power]\label{ass:inputConstr}
The expected average power of the (potentially random) input sequence $\{u(t)\}_{t=0}^{T-1}$ is bounded by $\gamma_u^2$, i.e., $\Expect\left[\sum_{t=0}^{T-1} \Vert u(t)\Vert ^2\right] \le \gamma_u^2 T$.
\end{assumption}

\subsection{Sample complexity lower bounds} \label{sec:lowerBound}
To formalize the class of algorithms considered in this section, we define the class of $\delta$-correct algorithms.
\begin{definition}[$\delta$-correct algorithms]
    Consider the setup described in Section~\ref{sec:ProblemSetup}. An algorithm is called $\delta$-correct, if for all $\delta\in (0,1)$, and any $(\theta_*$, $\SSet)$ there exists a finite time $\bar T$ such that for all $t\ge \bar T$ the algorithm returns an estimate $\hat\theta_t$ that satisfies $\Prob[\hat \theta_t = \theta_*]\ge 1-\delta$. 
\end{definition}
Note that by restricting to $\delta$-correct algorithms, we can obtain instance-specific sample complexity lower bounds that hold for any reasonable algorithm and any input sequence. 
\begin{theorem}\label{th:SampleComplexLowerGeneral}
    Consider the unknown dynamical system~\eqref{eq:TrueSysEvo} with $x(0) = 0$  and the set $\SSet$ defined in~\eqref{eq:defSet}.
    Then for any (potentially random) excitation input sequence\footnote{Random and deterministic input sequences are denoted identically. The interpretation will be clear from the context.}  $U\in \R^{n_u \bar{T}}$
    and any $\delta$-correct algorithm it holds that
    \begin{align}
        \min_{i\in[1, N]} &\Expect\left[{U}^\top W_i(\bar{T}) U\right]   \label{eq:sampleComplexLowerGeneral} \\ &+  
        \tr{S_w(\bar{T})^\top  Q_{\Sigma_w^{-1}}^i(\bar{T}) S_w(\bar{T})}
        \ge 2 \log\left(\frac{1}{2.4 \delta}\right).     \notag
    \end{align}
    Furthermore, under Assumption~\ref{ass:inputConstr}, the lower bound is minimized by the excitation input  
    \begin{align}
        U^* \in \argmax_{U^\top U \le \gamma_u^2 \bar{T}} \min_{i\in [1,N]}  &U^\top W_i(\bar{T}) U \label{eq:defOptU} \\ & + 
        \mathrm{tr}\big({S_w(\bar{T})^\top Q_{\Sigma_w^{-1}}^i(\bar{T}) S_w(\bar{T})}\big), \notag
    \end{align}
    and when applying $U^*$ any $\delta$-correct algorithm satisfies
    \begin{align}
        &\min_{p\in \Delta_N} \gamma_u^2 \bar{T} \lambda_\mathrm{max} \Big(\sum_{i=1}^{N} p_i W_i(\bar{T})\Big) \label{eq:lowerBoundOptimal} \\ 
        &+ \max_{i\in [1,N]} \tr{S_w(\bar{T})^\top  Q_{\Sigma_w^{-1}}^i(\bar{T}) S_w(\bar{T})}
        \ge 2 \log\left(\frac{1}{2.4 \delta}\right). \notag
    \end{align}
\end{theorem}
While using the maximum of the trace term in \eqref{eq:lowerBoundOptimal} introduces some conservatism, in practical applications the noise is often significantly smaller than the input. In these cases~\eqref{eq:lowerBoundOptimal} is dominated by the first term and the added conservatism is small.
The proof of Theorem~\ref{th:SampleComplexLowerGeneral} is presented in Appendix~\ref{secAppendix:SampleComplexLowerGeneral}.
Note that while the solution to \eqref{eq:defOptU} might not be unique, any optimizer suffices to achieve the optimal lower bound. As shown in Appendix~\ref{app:solExists}, an optimal solution exists. 
\begin{remark}\label{rem:Lower}
    Recall from \eqref{eq:expectedErrorFinal}-\eqref{eq:defWi} that, given $x(0) = 0$, \eqref{eq:defOptU} is equivalent to 
    \begin{equation}
        U^* \in \argmax_{U^\top U \le \gamma_u^2\bar{T}} \min_{i \in [1, N]} \Expect\left[\tilde{\varepsilon}_{\theta_i}(0, \bar{T})\right].
    \end{equation}
    Thus, Theorem~\ref{th:SampleComplexLowerGeneral} implies that the optimal excitation $U^*$ maximizes $\Expect\left[\tilde{\varepsilon}_{\theta_i}(0, \bar{T})\right]$ uniformly over  $\theta_i \in \SSet\setminus \{\theta_*\}$. 
    Thus, $U^*$ maximizes the distance between $\Expect\left[\tilde{\varepsilon}_{\theta_*}(0, \bar{T})\right]=0$ and $\Expect\left[\tilde{\varepsilon}_{\theta_i}(0, \bar{T})\right]$, to separate the 
    true system $\theta_*$ from all other systems $\theta_i \in \SSet\setminus\{\theta_*\}$ as clearly as
    possible.
\end{remark}
Theorem~\ref{th:SampleComplexLowerGeneral} can be used to obtain a sample complexity lower bound for isotropic Gaussian inputs.  
\begin{corollary}
    \label{co:lowerBoundGaussian}
    Consider the unknown dynamical system~\eqref{eq:TrueSysEvo} with $x(0) = 0$
    and the set $\SSet$ defined in~\eqref{eq:defSet}. Suppose that $u(t) \simiid \N (0, \frac{\gamma_u^2}{n_u} I_{n_u})$. Then for any $\delta$-correct algorithm it holds that
    \begin{align}
             \min_{i\in [1, N]} & \gamma_u^2 \bar{T} \lambda_\mathrm{mean}\big(W_i(\bar{T})\big) \label{eq:lowerBoundGaus}  \\ &+  
            \tr{S_w(\bar{T})^\top Q_{\Sigma_w^{-1}}^i(\bar{T}) S_w(\bar{T})} 
            \ge 2\log\left(\frac{1}{2.4 \delta}\right). \notag
    \end{align}    
\end{corollary}
The proof is provided in Appendix~\ref{secAppendix:Proof_LowerBound_random}.
To gain an intuition for the results in this chapter consider the following example, where $\SSet$ consists of only two systems which eliminates the need for minimization.
\begin{example}\label{ex:Example1}
    First, consider the identification problem defined by
    \begin{equation}\label{eq:Example_first_sys}
        A_* = \begin{bmatrix}
            0 & 0.1 \\ 0 & 0
        \end{bmatrix} \quad 
        B_* = \begin{bmatrix}
            0 \\ 1
        \end{bmatrix} 
               \quad 
    A_1 = \begin{bmatrix}
        0 & 0.2 \\ 0 & 0
    \end{bmatrix}
    \end{equation}
    and $\SSet=\{(A_*, B_*), (A_1, B_*)\}$.
        A straightforward calculation yields $W_1(t) = \diag{t}\left(\begin{bmatrix} 0.01 \end{bmatrix}\right)$ and $\lambda_j(W_1(t)) =  0.01$, $\forall j \in [1, t], \forall t\ge1$. 
        Thus, in this case isotropic Gaussian excitations are optimal in the sense of the sample complexity lower bound. 
Next, consider $\tilde{\SSet}=\{(\tilde{A}_*, \tilde{B}_*), (\tilde{A}_1, \tilde{B}_*)\}$ with 
        \begin{equation}\label{eq:Example_second_sys}
        \begin{aligned}
            \tilde{A}_* &= \begin{bmatrix}
                A_* & 0_{2\times d} \\
                0_{d\times 2}  & I_d 
            \end{bmatrix}
            \quad 
            \tilde{B}_* = \begin{bmatrix}
                B_* & 0_{2\times d} \\
                0_{d\times 1}   & I_d
            \end{bmatrix}
            \\ 
            \tilde{A}_1 &= \begin{bmatrix}
                A_1 & 0_{2\times d} \\
                0_{d\times 2}   & I_d 
            \end{bmatrix}
            \quad 
            \tilde{B}_1 = \tilde{B}_*,
        \end{aligned}
        \end{equation}
        for some fixed $d>0$. 
Through direct calculations, we obtain $\tilde{W}_1(t) = \diag{t}\left(\begin{bmatrix} 0.01 & 0_{1\times d} \\ 0_{d \times 1} & 0_{d\times d}\end{bmatrix}\right)$, $\forall t\ge1$ and hence $\lambda_\mathrm{mean}(\tilde{W}_1(t)) = \frac{1}{d+1}\lambda_\mathrm{max}(\tilde{W}_1(t))$. Hence, when $d$ is large, the sample complexity lower bound for isotropic Gaussian excitations is significantly larger than for the optimal oracle excitation. 
This indicates that using isotropic Gaussian excitations requires more samples to achieve the same confidence $\delta$ if algorithms matching the lower bound are used. This can also be observed in the sample complexity upper bounds we present in the next section. 
Further, under the same process noise conditions the sample complexity lower bound of the identification problems $(\theta_*, \SSet)$ and $(\tilde{\theta}_*, \tilde{\SSet})$ is identical given the respective optimal oracle excitation is used for both instances, as $\lambda_\mathrm{max}(W_1(t)) \equiv \lambda_\mathrm{max}(\tilde{W}_1(t))$. 
Hence, in this sense the difficulty of the problems $(\theta_*, \SSet)$ and $(\tilde{\theta}_*, \tilde{\SSet})$ is identical.
This is due to the fact that the last $d$ modes of the system~\eqref{eq:Example_second_sys} are not relevant for the identification problem since they are decoupled from the unknown part. 
Confirming this intuition, $\tilde{U}^*$ only excites the system with the first input.
Crucially, our mathematical formulation enables the analysis of complex problem setups, where the optimal excitation and hardness cannot be determined intuitively.
\end{example}
\subsection{Sample complexity upper bounds} \label{sec:UpperBound}
In this section, we establish a modular framework to derive sample complexity upper bounds for identifying $\theta_*$ with high probability. 
Our proposed framework builds on the notion of \acl{PE} which is closely connected to the observations made in the previous section.
\subsubsection{Persistency of Excitation}
Recall that in Example~\ref{ex:Example1}, the goal of identifying the true system can be accomplished by exciting only the first mode of the system using only the first control input. Clearly, there exists no $c>0$ for which this input satisfies
\begin{equation}
    \Expect\Big[\sum_{t=0}^{\tau-1} \begin{bmatrix}
        x(t) \\ u(t)
    \end{bmatrix} \begin{bmatrix}
        x(t)^\top & u(t)^\top  
    \end{bmatrix}\Big] \succeq c I_{n_x+n_u},
\end{equation}
which is the \ac{PE} condition required in the case of an infinite hypothesis class (cf. \cite{Tsiamis2023}). 
In fact, it is only necessary to excite the parts of the system that carry uncertainty. Thus, we introduce a weaker \ac{PE} condition, which has been used as an assumption in a similar form by \cite{muehlebach2025}.
\begin{definition}[\Acf{PE}]\label{def:PE}
    Consider an input sequence $\{u(t)\}_{t=0}^{\tau-1}$, satisfying $\Expect[\sum_{t=0}^{\tau-1} \Vert u(t)\Vert^2]= \gamma_u^2 \tau$ for some $\gamma_u>0$. We say $\{u(t)\}_{t=0}^{\tau-1}$ is  persistently exciting for $(\theta_*, \SSet)$  
    if there exist coefficients $c_u(\tau)>0$, $c_w(\tau)>0$ such that for any $x(0) \in \R^{n_x}$ and $\theta_i \in \SSet\setminus\{\theta_*\}$ 
    \begin{equation}\label{eq:defPE}
        \frac{1}{\tau} \sum_{t=0}^{\tau-1} \Expect \left[\Vert \Delta A_i x(t) + \Delta B_i u(t) \Vert^2_{\Sigma_w^{-1}} \right] \ge c_u(\tau)  \gamma_u^2+ c_w(\tau).
    \end{equation} 
\end{definition} 
Note that the \ac{PE} coefficients $c_u$ and $c_w$ are problem-specific and depend on the block length $\tau$. 
By the sample complexity lower bound, the excitation input needs to satisfy Definition~\ref{def:PE} for some $\tau>0$ to guarantee identification of $\theta_*$ with high probability (see Remark~\ref{rem:Lower} and Theorem~\ref{th:inputDependentLowerBound}).
If not, there exists some $\theta_i\neq \theta_*$ which yields the same state trajectory as $\theta_*$ and hence is indistinguishable from $\theta_*$ given the data.
In the following, we show that isotropic Gaussian excitations and the optimal oracle excitation $U^*$ defined in~\eqref{eq:defOptU} satisfy Definition~\ref{def:PE}, i.e., satisfy \ac{PE}.
\begin{lemma}[\ac{PE} for isotropic Gaussian inputs] \label{le:PEgaussianExcitation}
    Consider the system~\eqref{eq:TrueSysEvo} and the set $\SSet$ defined in~\eqref{eq:defSet}. 
    Then $u(t) \simiid \N(0, \frac{\gamma^2}{n_u} I_{n_u})$ is \ac{PE}  for any block length $\tau>0$ with 
    \vspace{-1em}
        \begin{align}
            c_u^{\mathrm{rand}}(\tau) &= \min_{p\in\Delta_N}\lambda_\mathrm{mean}\Big(\sum_{i=1}^{N}p_i W_{i}(\tau)\Big), \label{eq:cu_rand}\\
            c_w(\tau) &= \min_{i\in[1, N]} \frac{1}{\tau} \tr{S_w(\tau)^\top Q_{\Sigma_w^{-1}}^{i}(\tau) S_w(\tau)}.
        \end{align}
\end{lemma}
\begin{lemma}[\ac{PE} for optimal oracle excitation]\label{le:PEoptimalExcitation}
    Consider the system~\eqref{eq:TrueSysEvo} and the set $\SSet$ defined in~\eqref{eq:defSet}. 
    Then 
    \begin{equation} \label{eq:oracleUTau}
    \begin{aligned}
        U^* \in \argmax_{U^\top U \le \gamma_u^2 \tau} &\min_{i\in [1,N]} \\ &\sum_{t=0}^{\tau} \Expect \left[\Vert \Delta A_i x(t) + \Delta B_i u(t) \Vert^2_{\Sigma_w^{-1}} \right] 
    \end{aligned}
    \end{equation}
    is \ac{PE} for block length $\tau$ with 
        \begin{align}
            c_u^{\mathrm{opt}}(\tau) &= \min_{p\in \Delta_N}\lambda_\mathrm{max}\Big(\sum_{i=1}^N p_i W_i(\tau)\Big),  \label{eq:cu_opt} \\
            c_w(\tau) &= \min_{i\in [1,N]}\frac{1}{\tau} \tr{S_w(\tau)^\top Q_{\Sigma_w^{-1}}^{i}(\tau) S_w(\tau)}.
        \end{align}
\end{lemma}
The proof of both results is given in Appendices~\ref{secAppendix:proof_PE_Gaussian} and~\ref{secAppendix:proof_PE_Opt}, respectively.
Importantly, the \ac{PE} coefficients $c_u^\mathrm{rand}(\tau)$ and $c_u^\mathrm{opt}(\tau)$ coincide with the key problem parameters influencing the sample complexity lower bounds in Theorem~\ref{th:SampleComplexLowerGeneral} and Corollary~\ref{co:lowerBoundGaussian}.
Further, the optimal excitation $U^*$ derived in Theorem~\ref{th:SampleComplexLowerGeneral} maximizes the \ac{PE} coefficients $c_u(\tau)$, $c_w(\tau)$ in~\eqref{eq:defPE}.  
\begin{remark}\label{rem:Covariance}
    If $\Sigma_w$ is unknown and instead an estimate $\hat \Sigma_w$ or upper bound $\sigma_w^2I_{n_x} \succeq \Sigma_w$ is used to solve \eqref{eq:oracleUTau}, the corresponding input sequence is \ac{PE}, although with potentially suboptimal coefficients. 
    This sub-optimality is due to the lack of absence of information which could be used when designing excitation inputs.  
    The degree of sub-optimality is instance-dependent and can be analyzed by computing and comparing the coefficients. 
\end{remark}
%
%
\subsubsection{High probability identification with finite samples}
After establishing \ac{PE}, we now show that \ac{PE} guarantees fast identification of the true system~$\theta_*$ with high probability.
To do so, we introduce a sequential estimation algorithm for general input sequences.
In particular, we derive a sample complexity upper bound for Alg.~\ref{alg:ID} that holds for any \ac{PE} input sequence.
For estimation, the algorithm evaluates the sum of weighted empirical one-step prediction errors of $\theta_i$ up to time $t$ defined as 
\begin{equation} \label{eq:defEpsilon}
    \varepsilon_{\theta_i}(t) \coloneqq \sum_{s = 0}^{t-1} \Vert x(s+1) - A_i x(s) - B_i u(s)\Vert_{\Sigma_w^{-1}}^2 . 
\end{equation}
Clearly, $\varepsilon_{\theta_i}(t)$ can be interpreted as the negative log-likelihood of system $i$ given the data collected from time $0$ to $t$. 
Thus, the termination criterion of Alg.~\ref{alg:ID} is equivalent to a log-likelihood test.
%
%
\begin{algorithm}[t]
    \caption{Sequential identification algorithm}\label{alg:ID}
    \begin{algorithmic}[1]
    \Require $\SSet$, epoch length $\tau$, desired confidence $\delta$
    \For{$k = 1, 2, \dots$} 
    \State Collect data using \ac{PE} excitation input $\{u(t)\}_{t=(k-1)\tau}^{k\tau-1}$ with coefficients $c_{u_k}(\tau)$, $c_{w}(\tau)$
    \State Compute $\varepsilon_{\theta_i}(k\tau)$ for all $i \in [0, N]$
    \If{$\exists \hat \theta \in \SSet: \varepsilon_{\theta_i}(k\tau) - \varepsilon_{\hat \theta}(k\tau) > 2\log\left(\frac{N}{\delta}\right)$ for all $\theta_i \in \SSet\setminus \hat\theta$}
    \State Stop and \Return estimate $\hat \theta$
    \EndIf
    \EndFor
    \end{algorithmic}
\end{algorithm}
%
%
While related works considering this setup \citep{chatzikiriakos2024b,muehlebach2025} only analyze the case where the data is generated by  Gaussian excitations, the results in this work hold for \emph{any} input satisfying \ac{PE}. In particular, this allows us to systematically reason about how learning can be accelerated through particular choices of the excitation input.
To formalize this, we propose the following result.
\begin{theorem}\label{th:sampleComplexityUpperBound_General}
    Consider the unknown system~\eqref{eq:TrueSysEvo}, set $\SSet$ as defined in~\eqref{eq:defSet}. Then Alg.~\ref{alg:ID} yields an estimate $\hat{\theta}$ satisfying $\Prob\left[\hat \theta \neq \theta_*\right] \le \delta$ and with probability at least $1-\delta$ terminates no later than when $k$ satisfies 
    \begin{equation}
        k  c_w(\tau) + \sum_{j=1}^{k}c_{u_j}(\tau) \gamma_u^2 \ge c'\log\left(\frac{N}{\delta}\right),
    \end{equation}
    where $c'$ is a constant influenced by the variance in $\tilde{\varepsilon}_{\theta_i}$. 
\end{theorem}
The proof of Theorem~\ref{th:sampleComplexityUpperBound_General} is presented in Appendix~\ref{app:ProofsampleComplexityUpperBound_Gaussian}, where we also present the full version of the result.  
If the input sequence in each epoch is \ac{PE}, Theorem~\ref{th:sampleComplexityUpperBound_General} guarantees that the risk of miss-specification decays exponentially in $k$.
Further, Theorem~\ref{th:sampleComplexityUpperBound_General} establishes a modular framework to derive sample complexity upper bounds for different excitations. That is, to provide a sample complexity upper bound for Alg.~\ref{alg:ID} it is only necessary to show the excitation satisfies Definition~\ref{def:PE}. 
Given the \ac{PE} coefficients $c_{u_j}(\tau)$, $c_w(\tau)$ the sample complexity upper bound follows immediately. 
Clearly, the level of \ac{PE} of the excitation dictates the speed of identification, where larger \ac{PE} coefficients result in faster identification.
For the input sequences analyzed previously, the following result holds.  
\begin{corollary}\label{co:sampleComplexityRandomOracle}
    Consider the same setup as in Theorem~\ref{th:sampleComplexityUpperBound_General}.
    If $u(t)\simiid\N(0, \frac{\gamma_u^2}{n_u}I_{n_u})$ then Alg.~\ref{alg:ID} yields an estimate $\hat{\theta}$ satisfying $\Prob\big[\hat \theta \neq \theta_*\big] \le \delta$ and with probability at least $1-\delta$ terminates at the latest when $k$ first satisfies 
    \vspace{-1em}
    \begin{align}
        k\Big(&\gamma_u^2 \min_{p\in\Delta_N}\lambda_\mathrm{mean}\Big(\sum_{i=1}^{N}p_i W_{i}(\tau)\Big)  \\ &+ 
        \min_{i \in [1,N]} \frac{1}{\tau}\tr{S_w(\tau)^\top Q_{I_{n_x}}^{i} S_w(\tau)}\Big)
        \ge c' \log\left(\frac{N}{\delta}\right). \notag
    \end{align}
    Further, if the optimal oracle excitation input $U^*$ defined in \eqref{eq:oracleUTau} is applied the estimate $\hat{\theta}$ satisfies $\Prob[\hat \theta \neq \theta_*] \le \delta$, and with probability at least $1-\delta$ Alg.~\ref{alg:ID} terminates at the latest when $k$ first satisfies 
    \vspace{-0.3em}
    \begin{align}
        k \Big(&\gamma_u^2 \min_{p\in \Delta_N}\lambda_\mathrm{max}\Big(\sum_{i=1}^N p_i W_i(\tau)\Big) \\
        &+ \min_{i\in [1, N]}\frac{1}{\tau} \tr{S_w(\tau)^\top Q_{I_{n_x}}^{i} S_w(\tau)}
        \Big) 
        \ge c'' \log\left(\frac{N}{\delta}\right), \notag
    \end{align}
        where $c'$ and $c''$ are constants influenced by the variance in $\tilde{\varepsilon}_{\theta_i}$. 
\end{corollary}
After leveraging \ac{PE}, Corollary~\ref{co:sampleComplexityRandomOracle} follows directly from Theorem~\ref{th:sampleComplexityUpperBound_General}. We provide the proof and the full version in Appendix~\ref{secAppendix:Proofs_upperBound_Gaussian_Oracle}.
Note that since \ac{PE} coefficients~\eqref{eq:cu_rand},\eqref{eq:cu_opt} coincide with the key parameters entering the sample complexity lower bounds, the sample complexity upper bounds depend on the same parameters as the lower bounds. 
Thus, the lower and upper bounds qualitatively match, suggesting that they capture the key quantities contributing to the sample complexity. 
The epoch length $\tau$ can be used as a degree of freedom to optimize the sample complexity bound. 
Comparing Corollaries~\ref{co:lowerBoundGaussian} and ~\ref{co:sampleComplexityRandomOracle}, we see that, loosely speaking, whenever the matrices $W_i$ are ill-conditioned, input design can provide significant benefits in terms of the sample complexity. On the other hand, when all eigenvalues are identical, isotropic Gaussian excitation is already optimal and cannot be improved upon (cf. Example~\ref{ex:Example1}).
Notably, the matrices $W_i$ are often ill-conditioned in settings where prior knowledge is available. Hence, in settings where the uncertainty only lies in certain directions, input design can yield large benefits, which confirms intuition. 
If, however, the uncertainty is unstructured, the benefits of experiment design are significantly smaller since this yields matrices $W_i$ which are well-conditioned.

%
%
\section{Sequential input design algorithm}\label{sec:InputDesignAlgo}
In the previous section, we established that the matrices $W_i(\tau)$ are the key quantities that determine the degree of sub-optimality of random excitations compared to the optimal oracle excitation.
However, since the optimal oracle excitation input depends on $\theta_*$  it cannot be computed in practice.
This phenomenon is well known in the literature of experiment design and is usually tackled by using sequential algorithms to generate approximately optimal excitations.
%
%
\begin{algorithm}[t]
    \caption{Input design subroutine}\label{alg:inputDesign}
    \begin{algorithmic}[1]
    \Require $\SSet$, epoch length $\tau$, prediction errors $\varepsilon_{\theta_i}((k-1)\tau)$, state $x$, scaling $\rho_k\in [0, 1]$
    \State Compute weights using exponential weighting $w_{k+1}(i) = \exp\left(-\eta \varepsilon_{\theta_i}((k-1)\tau)\right)$ 
    \State Sample $\hat i_k \sim p_k(i)$, where $p_k(i) = w_k(i) / \sum_{j = 0}^{N} w_k(j)$ 
    \State Set $\hat \theta_k =\theta_{\hat{i}_k}$ and compute optimal input sequence $U_{\hat \theta_{k}}^*(x)$ by solving \eqref{eq:runningOptU}
    \State Define excitation according to $\rho_k$: $u_k^*(t) = \sqrt{1-\rho_k} u_{\hat \theta_k}^*(t) + \sqrt{\rho_k} u_\eta(t)$, $u_\eta(t) \simiid \N(0, \frac{\gamma_u^2}{n_u} I_{n_u})$
    \Return excitation sequence $u^*_k$ 
    \end{algorithmic}
\end{algorithm}
%
%
In line with these approaches, our proposed input design subroutine, which is presented as Alg.~\ref{alg:inputDesign}, can be called by Alg.~\ref{alg:ID} in each epoch and returns an excitation input sequence. To derive the excitation, Alg.~\ref{alg:inputDesign} draws an estimate $\hat \theta_k$ using exponential weights and solves the optimization~\eqref{eq:oracleUTau} using \ac{CE}, i.e., using $\hat \theta_k$ as if it were the true system. 
Thus, after sampling $\hat\theta_k = (\hat A_k, \hat B_k)$ and defining $\Delta \hat A_i \coloneqq \hat A_k - A_i$ and $\Delta \hat B_i \coloneqq \hat B_k - B_i$ the optimization problem 
\begin{align}
        U_{\hat{\theta}_k}^*(x(k\tau))  &\in \argmax_{U^\top U \le \gamma_u^2 \tau} \min_{\theta_i \neq \hat{\theta}_k } \label{eq:runningOptU}\\ &\Expect\Big[\hspace{-5pt}\sum_{t=k\tau}^{\tau(k+1)-1}\Vert \Delta \hat A_i x(t) + \Delta \hat B_i u(t) \Vert_{\Sigma_w^{-1}}^2\Big] \notag
\end{align}
is solved to compute the excitation input.
While this optimization problem uses knowledge of $\Sigma_w$, this is not required by our algorithm, and Alg.~\ref{alg:inputDesign} still provides PE inputs even when only an upper bound or no knowledge of $\Sigma_w$ is available (cf. Remark~\ref{rem:Covariance}).
We note that the optimization problem~\eqref{eq:runningOptU} is generally non-convex and, especially for $N \gg 0 $ might be challenging to solve. 
Since the focus of this work is statistical, we defer finding scalable approximations, as done by \cite{chatzikiriakos2025convex} for a similar setting, to future research. 
Due to the uncertainty in the estimate, we add random excitations to the optimal excitation input~\eqref{eq:runningOptU} according to a scaling parameter $\rho_k$, which can be adapted over time. 
We will analyze the choice of $\rho_k$ and its consequences on the speed of identification throughout this section.
Note that any $\rho_k \in [0, 1]$ yields excitation inputs satisfying Assumption~\ref{ass:inputConstr}.
As shown above, establishing sample complexity upper bounds reduces to showing that Alg.~\ref{alg:inputDesign} produces \ac{PE} inputs.  
\begin{lemma}\label{le:PEInputDesgin}
    Consider the system~\eqref{eq:TrueSysEvo}, set $\SSet$ as defined in \eqref{eq:defSet}. Consider the excitation input sequence $u_k^*$ generated by Alg.~\ref{alg:inputDesign} at some fixed iterate $k$ and let $\Prob[\hat \theta_k \neq \theta_*] \le p_k \in [0, 1]$, where $\hat \theta_k$ is the estimate drawn using exponential weights in round $k$.  
    Then $u_k^*$ satisfies \ac{PE} with 
    \begin{subequations}\label{eq:PEALgo}
        \begin{align}
            c_{u_k}^{\mathrm{Alg}}(\tau) &= (1-p_k)(1-\rho_k) c_u^\mathrm{opt}(\tau) + \rho_k c_u^\mathrm{rand}(\tau), \label{eq:cu_algo}\\
            c_{w}(\tau)&= \min_{i\in[1,N]}\frac{1}{\tau} \tr{S_w(\tau)^\top Q_{\Sigma_w^{-1}}^{i}(\tau) S_w(\tau)},
        \end{align}
    \end{subequations}
    where $c_u^\mathrm{opt}(\tau)$ is the PE coefficient for optimal excitation as defined in \eqref{eq:cu_opt} and $c_u^{\mathrm{rand}}(\tau)$ is the PE coefficient for isotropic Gaussian excitations as defined in \eqref{eq:cu_rand}.   
\end{lemma}
The proof of Lemma~\ref{le:PEInputDesgin} can be found in Appendix~\ref{secAppendix:proof_PE_Algo}.   
Note that while it might seem intuitive that any convex combination of two \ac{PE} inputs is also \ac{PE} this is not necessarily true. Given the setup in Example~\ref{ex:Example1} with $d=1$ it is easy to see that $u_1^\top(t) = \begin{bmatrix}
    1 & 0 
\end{bmatrix}$, $\forall t\ge0$ and $u_2^\top(t) = \begin{bmatrix}
    -1 & 0 
\end{bmatrix}$, $\forall t\ge0$ are both \ac{PE}. However, $u(t) = 0.5(u_1(t) + u_2(t)) = 0$ is clearly not \ac{PE}.
With Lemma~\ref{le:PEInputDesgin} we can directly obtain the following sample complexity upper bound. 
\begin{theorem}\label{th:sampleComlexityAlgo} 
    Consider Alg.~\ref{alg:ID} where the excitation input is derived by Alg.~\ref{alg:inputDesign}. Then, the estimate $\hat \theta$ of Alg.~\ref{alg:ID} satisfies $\Prob\left[\hat \theta \neq \theta_*\right] \le \delta$ and terminates no later than when $k$ satisfies
    \vspace{-0.3em}
    \begin{equation}
        k c_w(\tau) + \sum_{j= 1}^{k} c_{u_j}^\mathrm{Alg}(\tau) \ge c'\log\left(\frac{N}{\delta}\right),
    \end{equation}
    where $c_{u_j}^\mathrm{Alg}(\tau)$, $c_w(\tau)$ are defined in~\eqref{eq:PEALgo} and $c'$ is a constant influenced by the variance in $\tilde{\varepsilon}_{\theta_i}$. 
\end{theorem}
The full version of Theorem~\ref{th:sampleComlexityAlgo} and its proof are presented in Appendix~\ref{appSection:UpperAlgo}.
In essence, Lemma~\ref{le:PEInputDesgin} specifies that whenever the uncertainty is small enough compared to the expected benefit in \ac{PE} from the optimal excitation, \ac{CE} will provably improve the \ac{PE} coefficients compared to random Gaussian excitations. Theorem~\ref{th:sampleComlexityAlgo} shows that experiment design using \ac{CE} provably improves the sample complexity of identification in these cases.
Building on these arguments, the following result shows that as $\delta \to 0$ the proposed Algorithm exhibits a sample complexity upper bound that matches the lower bound up to constants.  
\begin{theorem} \label{th:asymptOptMain}
Consider Alg.~\ref{alg:ID} where the excitation input is derived by Alg.~\ref{alg:inputDesign}, with $\eta$ small enough. 
Let further $\rho_0 = c_0 \in (0, 1)$, $\lim_{k\to \infty}\rho_k = c_\infty \in [0, c_0]$, $\rho_{k-1} \ge \rho_k \ge (1-\varepsilon) \rho_{k-1}$ where $\varepsilon$ is a small enough non-negative constant.
Then for any $\delta \in (0,1)$ there exists a $k_\delta$, such that with probability $1-\delta$ Alg.~\ref{alg:ID} terminates after at most $k_\delta$ epochs with an estimate $\hat \theta$ that satisfies $\Prob[\hat \theta \neq \theta_*] \le \delta$ 
and 
\vspace{-1em}
\begin{align}
    \lim_{\delta \to 0} \frac{c' \log\left(\frac{N}{\delta}\right)}{k_\delta} &\le  \gamma_u^2 ( 1-c_\infty) \min_{p\in \Delta_N}\lambda_\mathrm{max}\Big(\sum_{i=1}^N p_i W_i(\tau)\Big) \notag \\ &+ \min_{i\in[1,N]}\frac{1}{\tau} \tr{S_w(\tau)^\top Q_{\Sigma_w^{-1}}^{i}(\tau) S_w(\tau)} 
\end{align}
where $c'$ is a constant influenced by the variance in $\tilde \varepsilon_{\theta_i}$.
\end{theorem}
The full version of Theorem~\ref{th:asymptOptMain} and its proof are presented in Appendix~\ref{app:asmptOpt}.
This result shows that, Alg.~\ref{alg:inputDesign} recovers the sample complexity lower bound up to constants as $\delta \to 0$ (cf. Theorem~\ref{th:SampleComplexLowerGeneral}).
\begin{remark}
Note that we need to resort to a blocking technique to show the result since the data is correlated. This has two consequences.
First, compared to Theorem~\ref{th:SampleComplexLowerGeneral} the upper bound depends on $W_i(\tau)$ instead of $W_i(\bar{T})$, where $\tau \ll \bar{T}$. 
For stable systems $W_i(\tau)$ converges quickly as $\tau$ increases, i.e., we have $W_i(\tau) \approx W_i(\bar{T})$ even for small $\tau$. 
However, if the system ~\eqref{eq:TrueSysEvo} is unstable, $W_i(\tau)$ can grow unbounded. Hence, while the upper bound remains valid for unstable systems, the gap to the lower bound is larger. 
Second, the sample complexity upper bound depends on the number of epochs $k$ instead of the total number of samples $T = k\tau$. Since the epoch length $\tau$ is usually a small integer, this only introduces a small constant factor.
\end{remark}
Further, as shown by Proposition~\ref{prop:decrease}, if $\rho_k$ decreases slow enough the \ac{PE}-coefficient of Alg.~\ref{alg:inputDesign}, $c_{u_k}^\mathrm{Alg}(\tau)$, increases as $k$ increases. 
Hence, Alg.~\ref{alg:inputDesign} outperforms random excitations after a finite number of iterations and, according to Proposition~\ref{prop:convergence}, up to constants converges to the optimal \ac{PE} coefficients.
Note that while $\varepsilon$ in Theorem~\ref{th:asymptOptMain} depends on oracle knowledge, any decrease that is slower than exponential can be used in practice, so that Theorem~\ref{th:asymptOptMain} guides the practical selection of $\rho_k$.
In particular, by Theorem~\ref{th:asymptOptMain}, even selecting a constant $\rho_k \in (0, 1)$ (e.g., $\rho_k \equiv \frac12$) will only impact the optimality by a constant as $\delta \to 0$.
We numerically evaluate how different selections of $\rho_k$ influence identification in Section~\ref{sec:numerics}. 
Interestingly, in our numerical studies, selecting $\rho_k=0$, i.e., greedily using \ac{CE},  yielded the best results. 
\begin{manualexample}{3.1 (continued)}
    Recall that in Example~\ref{ex:Example1} knowledge of the true system is not necessary to select the optimal excitation input sequence.  
    This is due to the structure in $\SSet$, which yields greedy \ac{CE}-based experiment design optimal (see numerical evaluation in Appendix~\ref{appSection:numerics}). 
   Deriving properties of $\SSet$ causing this is an interesting topic for future research. 
\end{manualexample}
%

\section{Numerical experiments}\label{sec:numerics}
In this section, we provide two numerical evaluations that showcase when active learning is beneficial and when its benefits are significantly smaller.\footnote{All experiments are carried out in Python using a standard notebook.  
The code for the numerical example can be accessed at: \url{https://github.com/col-tasas/2025-high-effort-low-gain}
}
First, we consider an example where $\vert \SSet \vert$ is small and exhibits structure. In a second example, we consider the case where $\SSet$ is larger and generated randomly.
While in the first case, active learning provides significant benefits in terms of identification speed, these benefits are significantly smaller in the second case. 
\\
\emph{Active learning can be high gain:}
Consider the set $\SSet = \{(A_*, B_*), (A_1, B_*), (A_2, B_*), (A_3, B_*)\}$ with $B_*^\top = \begin{bmatrix}
    0_{2\times 1} & I_2
\end{bmatrix}$ and
\begin{align*}
    A_* &= \begin{bmatrix}
        0 & 0.1 & 0 \\ 0 & 0 & 0 \\ 0 & 0 & 0.9
    \end{bmatrix} \quad  
    A_1 = \begin{bmatrix}
        0 & 0 & 0.1 \\ 0 & 0 & 0 \\ 0 & 0 & 0.9
    \end{bmatrix} 
    \\ 
    A_2 &= \begin{bmatrix}
        0 & 0 & 0.1 \\ 0 & 0 & 0 \\ 0 & 0 & 0.8
    \end{bmatrix}
    \quad 
    A_3 = \begin{bmatrix}
        0 & 0.1 & 0 \\ 0 & 0 & 0 \\ 0 & 0 & 0.8
    \end{bmatrix}.
\end{align*}
The noise covariance and input power are given by $\Sigma_w = I_{n_x}$ and $\gamma_u = 1$, respectively. 
For the numerical evaluation, we consider Alg.~\ref{alg:ID} where the excitation input is computed using Alg.~\ref{alg:inputDesign} with $\rho_k \equiv 0$ and $\eta = 0.01$, and compare it to the following data collection strategies used with Alg.~\ref{alg:ID}: 
1) Alg.~\ref{alg:inputDesign} with oracle knowledge (in each epoch \eqref{eq:runningOptU} is solved using $\theta_*$), 2) the optimization criterion proposed by \cite{wagenmaker2020active} with oracle knowledge, 3) $u(t)\simiid\N(0, \frac12 I_{n_u})$, and 4) optimal oracle excitation for the considered data length (fully offline).
We conduct 100 Monte Carlo simulations with 5 epochs and an epoch length of $15$ time steps.
In Figure~\ref{fig:numerics1} we display the mean and $\frac{\sigma}{2}$-band
of the likelihood of $\theta_*$ given the data over the epochs for the different data collection procedures. We investigate the likelihood since it is directly linked to the termination criterion of Alg.~\ref{alg:ID} and provides more insights than the termination times.
It is easy to see that the optimal oracle excitation performs best. 
To see why this is the case, note that the uncertainty in $\SSet$ is structured and only affects particular directions in the state space.
Hence, according to the theoretical insights presented in the previous sections, this setting admits large benefits of active learning. Consequently, isotropic Gaussian excitations perform significantly worse compared to all other input sequences.
While the active learning criterion by~\cite{wagenmaker2020active}, which does not account for the finite hypothesis class, outperforms isotropic Gaussian excitations, it falls short of the criterion presented in this work even when using oracle knowledge. 
This highlights the importance of leveraging knowledge of the finite hypothesis class in input design, when it is available.\\
\begin{figure*}[ht]
    \centering
    \begin{subfigure}{0.48\textwidth}
        \centering
        \resizebox{0.95 \linewidth}{!}{
%
%
\definecolor{mycolor1}{rgb}{0.000,0.447,0.698}
\definecolor{mycolor2}{rgb}{0.902,0.624,0.000}
\definecolor{mycolor3}{rgb}{0.000,0.620,0.451}
\definecolor{mycolor4}{rgb}{0.835,0.369,0.000}
\definecolor{mycolor5}{rgb}{0.35,0.35,0.35}
\begin{tikzpicture}

\begin{axis}[%
width=3.15in,
height=2.76in,
at={(0in,0in)},
scale only axis,
xmin=0,
xmax=5,
xtick = {0, 1, 2, 3, 4, 5},
xlabel style={font=\color{white!15!black}},
xlabel={Epoch},
ymin=0,
ymax=1,
ylabel style={font=\color{white!15!black}},
ylabel={Likelihood of $\theta_*$},
axis background/.style={fill=white},
axis x line*=bottom,
axis y line*=left,
xmajorgrids,
ymajorgrids,
legend style={at={(0.992,0.008)}, anchor=south east, legend cell align=left, align=left, draw=white!15!black}
]

\addplot [color=mycolor1, line width=2.0pt]
  table[row sep=crcr]{%
0	0.25\\
1	0.593662983877505\\
2	0.850479966278957\\
3	0.964957347459639\\
4	0.97765095399771\\
5	0.992594921475474\\
};
\addlegendentry{Alg.~\ref{alg:inputDesign} with $\rho_{k}\equiv0$}

\addplot[area legend, draw=none, fill=mycolor1, fill opacity=0.1, forget plot]
table[row sep=crcr] {%
x	y\\
0	0.25\\
1	0.43167919347573\\
2	0.712571989053695\\
3	0.896674116251954\\
4	0.909552353121407\\
5	0.960290460176005\\
5	1.02489938277494\\
4	1.04574955487401\\
3	1.03324057866732\\
2	0.98838794350422\\
1	0.75564677427928\\
0	0.25\\
}--cycle;

\addplot [color=mycolor2, dashed, line width=2.0pt]
  table[row sep=crcr]{%
0	0.25\\
1	0.593832623509608\\
2	0.849531854705965\\
3	0.964697382283808\\
4	0.977424927789187\\
5	0.992519064086355\\
};
\addlegendentry{Alg.~\ref{alg:inputDesign} with oracle}

\addplot[area legend, draw=none, fill=mycolor2, fill opacity=0.1, forget plot]
table[row sep=crcr] {%
x	y\\
0	0.25\\
1	0.431708271583804\\
2	0.711295329542226\\
3	0.896784798174069\\
4	0.90960490038677\\
5	0.960345566806842\\
5	1.02469256136587\\
4	1.0452449551916\\
3	1.03260996639355\\
2	0.987768379869705\\
1	0.755956975435412\\
0	0.25\\
}--cycle;

\addplot [color=mycolor3, line width=2.0pt]
  table[row sep=crcr]{%
0	0.25\\
1	0.328375809318214\\
2	0.457619389814656\\
3	0.553743089843619\\
4	0.704387711707722\\
5	0.755756125767002\\
};
\addlegendentry{Wagenmaker with oracle}

\addplot[area legend, draw=none, fill=mycolor3, fill opacity=0.1, forget plot]
table[row sep=crcr] {%
x	y\\
0	0.25\\
1	0.241477757503435\\
2	0.319438328709398\\
3	0.395079366168981\\
4	0.54310697441125\\
5	0.590372428484207\\
5	0.921139823049798\\
4	0.865668449004193\\
3	0.712406813518256\\
2	0.595800450919914\\
1	0.415273861132993\\
0	0.25\\
}--cycle;

\addplot [color=mycolor4, dashed, line width=2.0pt]
  table[row sep=crcr]{%
0	0.25\\
1	0.338238938626747\\
2	0.419202495853384\\
3	0.492159891453074\\
4	0.581043182963695\\
5	0.62481933600152\\
};
\addlegendentry{$u(t)\simiid \N(0, \frac{\gamma^2_u}{n_u}I_{n_u})$}

\addplot[area legend, draw=none, fill=mycolor4, fill opacity=0.1, forget plot]
table[row sep=crcr] {%
x	y\\
0	0.25\\
1	0.244857534808697\\
2	0.291908527019281\\
3	0.342907502637675\\
4	0.426752781586236\\
5	0.462021335040306\\
5	0.787617336962735\\
4	0.735333584341154\\
3	0.641412280268472\\
2	0.546496464687487\\
1	0.431620342444798\\
0	0.25\\
}--cycle;

\addplot [color=mycolor5, line width=2.0pt]
  table[row sep=crcr]{%
0	0.25\\
1	0.593832623509608\\
2	0.878409721932303\\
3	0.984429497214313\\
4	0.981119439536024\\
5	0.990089380648476\\
};
\addlegendentry{Optimal oracle excitation}

\addplot[area legend, draw=none, fill=mycolor5, fill opacity=0.1, forget plot]
table[row sep=crcr] {%
x	y\\
0	0.25\\
1	0.431708271583804\\
2	0.750925276802954\\
3	0.941536871004465\\
4	0.914815090995688\\
5	0.944236553900371\\
5	1.03594220739658\\
4	1.04742378807636\\
3	1.02732212342416\\
2	1.00589416706165\\
1	0.755956975435412\\
0	0.25\\
}--cycle;

\end{axis}
\end{tikzpicture}%
} 
        \caption{Structured set $\SSet$.}
        \label{fig:numerics1}
    \end{subfigure}\hfill
    \begin{subfigure}{0.48\textwidth}
        \centering
            \resizebox{0.95 \linewidth}{!}{
%
%
\definecolor{mycolor1}{rgb}{0.000,0.447,0.698}
\definecolor{mycolor2}{rgb}{0.902,0.624,0.000}
\definecolor{mycolor3}{rgb}{0.000,0.620,0.451}
\definecolor{mycolor4}{rgb}{0.835,0.369,0.000}
\definecolor{mycolor5}{rgb}{0.35,0.35,0.35}
\definecolor{mycolor6}{rgb}{0.800,0.475,0.655}
\begin{tikzpicture}

\begin{axis}[%
width=3.15in,
height=2.76in,
at={(0in,0in)},
scale only axis,
xmin=0,
xmax=5,
ymin=0,
ymax=1.01,
ytick={0.2 , 0.4, 0.6, 0.8, 1},
xtick={0, 1, 2, 3, 4, 5, 6, 7, 8},
ylabel style={font=\color{white!15!black}},
xlabel={Epoch},
ylabel={Likelihood of $\theta_*$},
axis background/.style={fill=white},
axis x line*=bottom,
axis y line*=left,
xmajorgrids,
ymajorgrids,
legend style={at={(0.992,0.008)}, anchor=south east, legend cell align=left, align=left, draw=white!15!black}
]
\addplot [color=mycolor1, line width=2.0pt]
  table[row sep=crcr]{%
0	0.0499999999999998\\
1	0.697277079413502\\
2	0.938002703278803\\
3	0.983345849553261\\
4	0.999537402793445\\
5	0.999999620793387\\
};
\addlegendentry{Alg.~\ref{alg:inputDesign} with $\rho_{k}\equiv0$}

\addplot[area legend, draw=none, fill=mycolor1, fill opacity=0.1, forget plot]
table[row sep=crcr] {%
x	y\\
0	0.0499999999999999\\
1	0.533664294353774\\
2	0.845281153622994\\
3	0.931173810157194\\
4	0.997514390328823\\
5	0.999998325829389\\
5	1.00000091575739\\
4	1.00156041525807\\
3	1.03551788894933\\
2	1.03072425293461\\
1	0.86088986447323\\
0	0.05\\
}--cycle;

\addplot [color=mycolor2, dashed, line width=2.0pt]
  table[row sep=crcr]{%
0	0.0499999999999998\\
1	0.754521579322641\\
2	0.953351059086705\\
3	0.98203572563215\\
4	0.991326551015068\\
5	0.999999972530883\\
};
\addlegendentry{Alg.~\ref{alg:inputDesign} with oracle}

\addplot[area legend, draw=none, fill=mycolor2, fill opacity=0.1, forget plot]
table[row sep=crcr] {%
x	y\\
0	0.0499999999999999\\
1	0.602324620977852\\
2	0.862411525124137\\
3	0.918863473099241\\
4	0.947960509738573\\
5	0.999999873674596\\
5	1.00000007138717\\
4	1.03469259229156\\
3	1.04520797816506\\
2	1.04429059304927\\
1	0.906718537667431\\
0	0.05\\
}--cycle;

\addplot [color=mycolor3, line width=2.0pt]
  table[row sep=crcr]{%
0	0.0499999999999998\\
1	0.291397816107454\\
2	0.727220246168472\\
3	0.947628949669908\\
4	0.993162829077197\\
5	0.990952601341275\\
};
\addlegendentry{Wagenmaker with oracle}

\addplot[area legend, draw=none, fill=mycolor3, fill opacity=0.1, forget plot]
table[row sep=crcr] {%
x	y\\
0	0.0499999999999999\\
1	0.187647008312921\\
2	0.577134032005501\\
3	0.867371723200049\\
4	0.975411399649941\\
5	0.9527049527423\\
5	1.02920024994025\\
4	1.01091425850445\\
3	1.02788617613977\\
2	0.877306460331442\\
1	0.395148623901987\\
0	0.05\\
}--cycle;

\addplot [color=mycolor4, dashed, line width=2.0pt]
  table[row sep=crcr]{%
0	0.0499999999999998\\
1	0.334939623397077\\
2	0.785670235132334\\
3	0.929046857961531\\
4	0.980169771212796\\
5	0.993405998590643\\
};
\addlegendentry{$u(t)\simiid \N(0, \frac{\gamma^2_u}{n_u}I_{n_u})$}

\addplot[area legend, draw=none, fill=mycolor4, fill opacity=0.1, forget plot]
table[row sep=crcr] {%
x	y\\
0	0.0499999999999999\\
1	0.207802865538935\\
2	0.642933603182516\\
3	0.832325735364122\\
4	0.924704722418521\\
5	0.963980206424515\\
5	1.02283179075677\\
4	1.03563482000707\\
3	1.02576798055894\\
2	0.928406867082151\\
1	0.462076381255218\\
0	0.05\\
}--cycle;

\addplot [color=mycolor5, line width=2.0pt]
  table[row sep=crcr]{%
0	0.0499999999999998\\
1	0.325466944903653\\
2	0.860077187914253\\
3	0.963949316031292\\
4	0.999054388634992\\
5	0.999994742890717\\
};
\addlegendentry{$\rho_{k}= \frac{\text{1}}{\text{1+k}}$}

\addplot[area legend, draw=none, fill=mycolor5, fill opacity=0.1, forget plot]
table[row sep=crcr] {%
x	y\\
0	0.0499999999999999\\
1	0.203681263400024\\
2	0.73262419996314\\
3	0.8918172631758\\
4	0.996381774572636\\
5	0.999974665319251\\
5	1.00001482046218\\
4	1.00172700269735\\
3	1.03608136888678\\
2	0.987530175865365\\
1	0.447252626407281\\
0	0.05\\
}--cycle;

\addplot [color=mycolor6, dashed, line width=2.0pt]
  table[row sep=crcr]{%
0	0.0499999999999998\\
1	0.295849608576334\\
2	0.878082677572471\\
3	0.986109221052825\\
5	0.999760675141641\\
};
\addlegendentry{$\rho_{k}= \frac{1}{(1+k)^2}$}

\addplot[area legend, draw=none, fill=mycolor6, fill opacity=0.1, forget plot]
table[row sep=crcr] {%
x	y\\
0	0.0499999999999999\\
1	0.165473971957759\\
2	0.757065624975272\\
3	0.93734302921943\\
4	0.95715642172153\\
5	0.998676083855273\\
5	1.00084526642801\\
4	1.02847910805317\\
3	1.03487541288622\\
2	0.99909973016967\\
1	0.426225245194909\\
0	0.05\\
}--cycle;

\end{axis}

\end{tikzpicture}%
            }
            \caption{Randomly generated set $\SSet$.}
            \label{fig:numericsRandomExample}
    \end{subfigure}
    \caption{Mean and $\frac{\sigma}{2}$-band of the likelihood of $\theta_*$ given the data for different excitation strategies and settings.}
    \label{fig:NumericsOverall}
\end{figure*}
\emph{Experiment design can be low gain:}
Now we consider a setup where the set $\SSet$ is generated randomly. 
To be precise, we consider the case where $\SSet$ consists of $20$ additional systems, where the parameters are drawn from a Gaussian around the true parameters $(A_*, B_*)$. The true system matrices $(A_*, B_*)$ and all other problem parameters are the same as in the previous section, to single out the effect structure in $\SSet$ has on the results. 
We again compare Alg.~\ref{alg:inputDesign} with $\rho_k\equiv 0$
to 1) Alg.~\ref{alg:inputDesign} with oracle knowledge (in each epoch \eqref{eq:runningOptU} is solved using $\theta_*$), 2) the optimization criterion proposed by \cite{wagenmaker2020active} with oracle knowledge, and 3) $u(t)\simiid\N(0, \frac12 I_{n_u})$. 
Further, we consider different time-varying choices of the weighting parameter $\rho_k$, namely $\rho_k = \frac{1}{1+k}$ and $\rho_k = \frac{1}{(1+k)^2}$.
The resulting mean and $\frac{\sigma}{2}$-band of the likelihood of $\theta_*$ given the data over 100 Monte Carlo simulations are displayed in Figure~\ref{fig:numericsRandomExample}.
Again, isotropic Gaussian excitations perform worse than all  excitation sequences generated by Alg.~\ref{alg:inputDesign}. 
Notice, however, that compared to the previous setting, the gap is significantly smaller. This is due to the fact that the set $\mathcal{S}$ consists of randomly generated systems. Hence, the set exhibits less structure, and uniform exploration is closer to optimal. This is in accordance with our theoretical results and indicates that there exist non-trivial cases where experiment design provides only a small benefit over isotropic random excitations. 
The criterion presented in~\cite{wagenmaker2020active} performs similarly to random excitations and worse than our approach for all considered choices of $\rho_k$.
Interestingly, in all our numerical evaluations, Alg.~\ref{alg:inputDesign} with $\rho_k \equiv 0$ performs on par with the online oracle approach, even though it lacks the knowledge of the true system and relies on \ac{CE}. 
In particular, this choice outperformed all other selections of $\rho_k$.
This means that, even though Alg.~\ref{alg:inputDesign} does not necessarily plan with the correct system initially, the obtained data is still informative, and in particular significantly more informative than data that is collected using isotropic Gaussian excitations. 
Capturing this effect mathematically and understanding how it depends on the structure of the set $\SSet$ is an interesting direction for future research.    
%
\vspace{-0.5em}
\section{Conclusion}
\vspace{-0.5em}
In this work, we analyze the problem of identifying an unknown linear dynamical system from a finite hypothesis class.
We present sample complexity lower bounds for isotropic Gaussian and arbitrary excitation, which can be used to determine the instance-specific benefit of experiment design.
We introduce a tailored notion of \ac{PE}, which gives rise to a modular framework to establish sample complexity upper bounds for any excitation input that satisfies \ac{PE}.
Based on our findings, we propose a \ac{CE}-based algorithm for experiment design and derive matching sample complexity upper bounds. 
Our analysis enables a quantitative assessment of the benefits provided by input design algorithms and identifies settings in which these benefits are limited, especially compared to the often considerable implementation complexity and computational effort of the underlying optimization problems. 
Numerical studies showcase that the proposed algorithm is highly competitive and achieves a performance close to the optimal oracle excitation.
We see extensions to systems with partially observed states, where the decay rate of the Markov parameters impacts the optimal excitation, and investigations of the case where the true system does not lie in $\SSet$ as interesting extensions for future research.

\clearpage
\section*{Acknowledgements}
Nicolas Chatzikiriakos and Andrea Iannelli acknowledge the support by Deutsche Forschungsgemeinschaft (DFG, German Research Foundation) under Germany's Excellence Strategy - EXC 2075 – 390740016 and by the Stuttgart Center for Simulation Science (SimTech).

\bibliography{references.bib}

@InProceedings{simchowitz2018learning,
  author       = {Simchowitz, M. and Mania, H. and Tu, S. and Jordan, M. I. and Recht, B.},
  booktitle    = {Conf. On Learning Theory},
  title        = {Learning without mixing: Towards a sharp analysis of linear system identification},
  year         = {2018},
  organization = {PMLR},
  pages        = {439--473},
}

@InProceedings{jedra2019sample,
  author    = {Jedra, Y. and Proutiere, A.},
  booktitle = {IEEE 58th Conf. on Decision and Control (CDC)},
  title     = {Sample complexity lower bounds for linear system identification},
  year      = {2019},
}

@Article{jedra2022finite,
  author    = {Jedra, Y. and Proutiere, A.},
  journal   = {IEEE Transactions on Automatic Control},
  title     = {Finite-time identification of linear systems: Fundamental limits and optimal algorithms},
  year      = {2022},
  publisher = {IEEE},
}

@InProceedings{wagenmaker2020active,
  author       = {Wagenmaker, A. and Jamieson, K.},
  booktitle    = {Conf. on Learning Theory},
  title        = {Active learning for identification of linear dynamical systems},
  year         = {2020},
  organization = {PMLR},
  pages        = {3487--3582},
}

@InProceedings{sarkar2019near,
  author       = {Sarkar, T. and Rakhlin, A.},
  booktitle    = {International Conf. on Machine Learning},
  title        = {Near optimal finite time identification of arbitrary linear dynamical systems},
  year         = {2019},
  organization = {PMLR},
  pages        = {5610--5618},
}

@InProceedings{Abbasi2011Improved,
  author    = {Abbasi-Yadkori, Y. and P., D\'{a}vid and Szepesv\'{a}ri, C.},
  booktitle = {Advances in Neural Information Processing Systems},
  title     = {Improved Algorithms for Linear Stochastic Bandits},
  year      = {2011},
  volume    = {24},
}

@Article{dean2020sample,
  author    = {Dean, S. and Mania, H. and Matni, N. and Recht, B. and Tu, S.},
  journal   = {Foundations of Computational Mathematics},
  title     = {On the sample complexity of the linear quadratic regulator},
  year      = {2020},
  number    = {4},
  pages     = {633--679},
  volume    = {20},
  publisher = {Springer},
}

@Book{wainwright2019high,
  author    = {Wainwright, M. J.},
  publisher = {Cambridge university press},
  title     = {High-dimensional statistics: A non-asymptotic viewpoint},
  year      = {2019},
  volume    = {48},
}

@Article{Garivier2019,
  author    = {Garivier, A. and Ménard, P. and Stoltz, G.},
  journal   = {Mathematics of Operations Research},
  title     = {Explore First, Exploit Next: The True Shape of Regret in Bandit Problems},
  year      = {2019},
  number    = {2},
  pages     = {377--399},
  volume    = {44},
}

@Article{Tsiamis2023,
  author    = {Tsiamis, A. and Ziemann, I. and Matni, N. and Pappas, G. J.},
  journal   = {IEEE Control Systems},
  title     = {Statistical Learning Theory for Control: A Finite-Sample Perspective},
  year      = {2023},
  number    = {6},
  pages     = {67--97},
  volume    = {43},
}

@inproceedings{tsiamis2021linear,
  title={Linear systems can be hard to learn},
  author={Tsiamis, Anastasios and Pappas, George J},
  booktitle={60th IEEE Conf. on Decision and Control (CDC)},
  year={2021},
  organization={IEEE}
}

@Article{Johnson2004,
  author    = {Johnson, J. and Omland, K.},
  journal   = {Trends in Ecology \& Evolution},
  title     = {Model selection in ecology and evolution},
  year      = {2004},
  issn      = {0169-5347},
  number    = {2},
  pages     = {101--108},
  volume    = {19},
  doi       = {10.1016/j.tree.2003.10.013},
  publisher = {Elsevier BV},
}

@Book{Lattimore2020,
  author    = {Lattimore, T. and Szepesv{\'a}ri, C.},
  publisher = {Cambridge University Press},
  title     = {Bandit algorithms},
  year      = {2020},
}

@Article{chatzikiriakos2024b,
      title={Sample Complexity Bounds for Linear System Identification from a Finite Set}, 
      author={Nicolas Chatzikiriakos and Andrea Iannelli},
      journal={IEEE Control Systems Letters}, 
      year={2024},
      volume={8},
      number={},
      pages={2751-2756},}

@article{muehlebach2025,
      title={The Sample Complexity of Online Reinforcement Learning: A Multi-model Perspective}, 
      author={Michael Muehlebach and Zhiyu He and Michael I. Jordan},
      year={2025},
      journal={arXiv:2501.15910},
      }

@InProceedings{wagenmaker21a,
  title = 	 {Task-Optimal Exploration in Linear Dynamical Systems},
  author =       {Wagenmaker, Andrew J and Simchowitz, Max and Jamieson, Kevin},
  booktitle = 	 {38th International Conf. on Machine Learning},
  pages = 	 {10641--10652},
  year = 	 {2021},
  publisher =    {PMLR},
  }

@article{Chernoff59Experiments,
author = {Herman Chernoff},
title = {{Sequential Design of Experiments}},
volume = {30},
journal = {The Annals of Mathematical Statistics},
number = {3},
publisher = {Institute of Mathematical Statistics},
pages = {755 -- 770},
year = {1959},
}

@article{kiefer1960equivalence,
  title={The equivalence of two extremum problems},
  author={Kiefer, Jack and Wolfowitz, Jacob},
  journal={Canadian Journal of Mathematics},
  volume={12},
  pages={363--366},
  year={1960},
  publisher={Cambridge University Press}
}

@article{Kiefer74optimumDesigns,
 author = {J. Kiefer},
 journal = {The Annals of Statistics},
 number = {5},
 pages = {849--879},
 publisher = {Institute of Mathematical Statistics},
 title = {General Equivalence Theory for Optimum Designs (Approximate Theory)},
 volume = {2},
 year = {1974}
}

@article{mania2022active,
  title={Active learning for nonlinear system identification with guarantees},
  author={Mania, Horia and Jordan, Michael I and Recht, Benjamin},
  journal={J. of Machine Learning Research},
  volume={23},
  number={32},
  pages={1--30},
  year={2022}
}

@INPROCEEDINGS{lee2024active,
  author={Lee, Bruce D. and Ziemann, Ingvar and Pappas, George J. and Matni, Nikolai},
  booktitle={IEEE 63rd Conf. on Decision and Control (CDC)}, 
  title={Active Learning for Control-Oriented Identification of Nonlinear Systems}, 
  year={2024},
  }

@article{rojas2007robust,
  title={Robust optimal experiment design for system identification},
  author={Rojas, Cristian R and Welsh, James S and Goodwin, Graham C and Feuer, Arie},
  journal={Automatica},
  volume={43},
  number={6},
  pages={993--1008},
  year={2007},
  publisher={Elsevier}
}

@inproceedings{gerencser2005adaptive,
  title={Adaptive input design in system identification},
  author={Gerencs{\'e}r, L{\'a}szl{\'o} and Hjalmarsson, H{\aa}kan},
  booktitle={44th IEEE Conf. on Decision and Control},
  pages={4988--4993},
  year={2005},
}

@book{1977dynamic,
  title={Dynamic System Identification: Experiment Design and Data Analysis},
  author={Graham Goodwin and Robert Payne},
  series={Mathematics in Science and Engineering},
  year={1977},
  publisher={Academic Press}
}

@article{bombois2011optimal,
  title={Optimal experiment design for open and closed-loop system identification},
  author={Bombois, Xavier and Gevers, Michel and Hildebrand, Roland and Solari, Gabriel},
  journal={Commun.  in Information and Systems},
  volume={11},
  number={3},
  pages={197--224},
  year={2011},
  publisher={Int. Press of Boston}
}

@InProceedings{sattar2022Bilinear,
  author={Sattar, Yahya and Oymak, Samet and Ozay, Necmiye},
  booktitle={IEEE 61st Conf. on Decision and Control (CDC)}, 
  title={Finite Sample Identification of Bilinear Dynamical Systems}, 
  year={2022},
  volume={},
  number={},
  pages={6705-6711},}

@article{sattar2022non,
  title={Non-asymptotic and accurate learning of nonlinear dynamical systems},
  author={Sattar, Yahya and Oymak, Samet},
  journal={Journal of Machine Learning Research},
  volume={23},
  number={140},
  pages={1--49},
  year={2022}
}

@InProceedings{foster20a,
  title = 	 {Learning nonlinear dynamical systems from a single trajectory},
  author =       {Foster, Dylan and Sarkar, Tuhin and Rakhlin, Alexander},
  booktitle = 	 {Proceedings of the 2nd Conf. on Learning for Dynamics and Control},
  pages = 	 {851--861},
  year = 	 {2020},
  volume = 	 {120},
  publisher =    {PMLR},
}

@article{rigollet2023highdimensionalstatistics,
      title={High-Dimensional Statistics}, 
      author={Philippe Rigollet and Jan-Christian Hütter},
      year={2023},
      journal={arXiv:2310.19244},
}

@inproceedings{chatzikiriakos2025convex,
  title={Hidden convexity in active learning: A convexified online input design for ARX systems},
  author={Chatzikiriakos, Nicolas and Song, Bowen and Rank, Philipp and Iannelli, Andrea},
  booktitle={IEEE 64th Conf. on Decision and Control (CDC)},
  pages={63--68},
  year={2025},
}

@inproceedings{miljkovic2011fault,
  title={Fault detection methods: A literature survey},
  author={Miljkovi{\'c}, Dubravko},
  booktitle={2011 Proceedings of the 34th international convention MIPRO},
  pages={750--755},
  year={2011},
  organization={IEEE}
}

\clearpage
\appendix
\thispagestyle{empty}

\onecolumn
\tableofcontents
%
%
\section{Related Works}\label{appSection:relatedWorks}
\textbf{Finite Sample System Identification}
Novel results in high-dimensional statistics \citep{Abbasi2011Improved, wainwright2019high} have sparked an increased interest in the finite sample analysis of system identification. 
Hereby, the sample complexity of identification is of particular importance and has been analyzed mostly for random excitations and for an infinite hypothesis class using the closed-form solution of the \ac{OLS} estimator.
Since the data collected from dynamical systems is highly correlated, the first work by~\cite{dean2020sample} only provided sample complexity upper bounds for the case where the data was collected from multiple trajectories.
Subsequent works were able to overcome this restrictive assumption and provided sample complexity upper bounds for marginally stable \citep{simchowitz2018learning} and also unstable linear dynamical systems \citep{sarkar2019near} for trajectory data. 
Hereby, it has been shown that the \ac{OLS} is statistically inconsistent when the data is collected from certain classes of unstable systems. 
Additional works extended these works to certain classes of nonlinear dynamical systems such as bilinear systems \citep{sattar2022Bilinear} generalized linear systems \citep{foster20a,sattar2022non}.
While sample complexity upper bounds are valuable since they establish identification error guarantees, sample complexity lower bounds are important to understand the hardness of the identification problem and to judge the tightness of the upper bounds. 
Building on information theoretic tools, \cite{jedra2022finite} provide sample complexity lower bounds that hold independently of the used algorithm. \cite{tsiamis2021linear} used similar tools to understand when learning is hard with isotropic Gaussian excitations.
All previous works consider the case where the hypothesis class is infinite. 
However, in many applications, there might exist some prior knowledge restricting the hypothesis class to be finite, i.e., there exists a finite set of systems the true system belongs to. 
This setup has been previously considered by \cite{chatzikiriakos2024b,muehlebach2025}, who observed that the lack of stability of the true system does not influence learning negatively, as is the case for the infinite hypothesis class. 
While the previously mentioned works consider several system classes and setups, they all assume the data is collected using (sub-)Gaussian inputs, which might be suboptimal in many cases.

\textbf{Experiment Design:}
The design of optimal experiments has a long history in learning theory \citep{Chernoff59Experiments, Kiefer74optimumDesigns, kiefer1960equivalence} where a large body of works exists concerning the question of how data should be sampled to obtain the most accurate estimate of an unknown quantity.  
One particular field of interest for experiment design is the identification of unknown dynamical systems, where, compared to the classical setup in learning theory, the data cannot be sampled arbitrarily but rather the unknown dynamical system governs the way the data can be collected (cf.  \cite{bombois2011optimal, 1977dynamic} for surveys on the topic). 
One key difficulty introduced by the dynamics of the unknown system is that the Fisher-information matrix naturally depends on the unknown system parameters. 
To tackle this problem, several algorithms have been proposed which can be broadly categorized into robust approaches that optimize for the worst-case using minmax objectives (cf. \cite{rojas2007robust} and the references therein) and sequential approaches that rely on a running estimate \cite{gerencser2005adaptive}. 
While classical works only consider the asymptotic case, recently several works have considered the problem of experiment design for the identification of an unknown dynamical system under a finite sample perspective.  
In particular, \cite{wagenmaker2020active,wagenmaker21a} consider \ac{LTI} systems and provide finite sample guarantees by leveraging sequential algorithms.
Further, \cite{chatzikiriakos2025convex} provide an exact convex reformulation of the design criterion presented by \cite{wagenmaker2020active}, enabling the efficient application of the algorithm. 
The proposed algorithm optimizes over periodic excitation signals using the \ac{CE} principle and augments the input with exploratory noise. 
Extending previous works \cite{lee2024active,mania2022active} consider certain classes of non-linear systems and provide finite sample guarantees for them. 
While optimality of the proposed experiment design algorithms is considered in the respective works, they do not analyze how large the benefit of experiment design is compared isotropic Gaussian excitations. 
Furthermore, the input space is restricted to periodic inputs, which, although shown to be sufficient by \cite{wagenmaker2020active}, is restrictive.
\section{Proofs for sample complexity lower bounds}\label{secAppendix:Proofs_lowerBound}
Before we proceed with presenting the proofs of the sample complexity lower bound results in Section~\ref{sec:lowerBound} we first provide the following input-dependent sample complexity result, which serves as the starting point for the subsequent results and is akin to \cite[Theorem 2]{jedra2019sample} which considers the case of an infinite hypothesis class instead of the finite hypothesis class considered in this work. 
\begin{theorem}[Input-dependent sample complexity lower bound]\label{th:inputDependentLowerBound}
Consider the unknown dynamical system~\eqref{eq:TrueSysEvo} and the hypothesis class $\SSet$ as defined in \eqref{eq:defSet}. Then any $\delta$-correct algorithm satisfies 
\begin{equation}
    \Expect \left[\sum_{t=0}^{\bar{T}-1}\Vert 
            \Delta A_i x(t) + \Delta B_i u(t) \Vert_{\Sigma_w^{-1}}^2 
         \right] \ge 2 \log\left(\frac{1}{2.4 \delta}\right)
\end{equation}
for any excitation input sequence $\{u(t)\}_{t=0}^{\bar{T}-1}$.
\end{theorem}
\begin{proof}
    Define the data as $\mathcal{D}_{\bar{T}} \coloneqq \{x(0), u(0), \dots, u(\bar{T}-1), x(\bar{T})\}$ and the probability of the observing $\mathcal{D}_{\bar{T}}$ under system $\theta_i$ as $\Prob_{\theta_i}(\mathcal{D}_{\bar{T}})$. Then, we define the log-likelihood ratio of the first $\bar{T}$ observations under $\theta_*$ and some $\theta_i \in \SSet\setminus\{\theta_*\}$ as 
    $L_{\bar{T}} = \log\left(\frac{\Prob_{\theta_*}(\mathcal{D}_{\bar{T}})}{\Prob_{\theta_i}(\mathcal{D}_{\bar{T}})}\right)$.
  Following the change of measurement argument presented by \cite{jedra2019sample}, we use the generalized data processing inequality \cite[Lemma 1]{Garivier2019} to obtain the lower bound 
    \begin{align*}
        \Expect\left[L_{\bar{T}}\right] &= \KL{\Prob_{\theta_*}(\mathcal{D}_{\bar{T}})}{\Prob_{\theta_i}(\mathcal{D}_{\bar{T}})} \\&\ge \sup_{\E \in \mathcal{F}_{\bar{T}}} \kl{\Prob_{\theta_*}(\E)}{\Prob_{\theta_i}(\E)}, 
    \end{align*}
    where $\kl{x}{y}$ is the KL-divergence of two Bernoulli distributions of means $x$ and $y$, respectively.
    Since we analyze $\delta$-correct algorithms we define the event $\E \coloneqq \{\hat \theta_{\bar{T}} = \theta_*\}$ which yields $\Prob_{\theta_*}(\E) \ge 1-\delta$ and $\Prob_{\theta_i}(\E) \le \delta$
    and hence 
    \begin{equation}
        \kl{\Prob_{\theta_*}(\E)}{\Prob_{\theta_i}(\E)} \ge (2\delta -1)\log\left(\frac{1-\delta }{\delta}\right) \ge \log\left(\frac{1}{2.4\delta}\right).
    \end{equation}
    Further, we follow \cite[Section IV.A]{jedra2019sample} to obtain 
    \begin{align*}
        \Expect\left[L_{\bar{T}}\right] &= \frac12 \Expect \left[\sum_{t=0}^{\bar{T}-1} \begin{bmatrix}
        x(t)^\top & u(t)^\top 
        \end{bmatrix} \begin{bmatrix}
        \Delta A_i^\top \\ \Delta B_i ^\top 
        \end{bmatrix} \Sigma_w^{-1} \begin{bmatrix}
            \Delta A_i & \Delta B_i 
        \end{bmatrix} \begin{bmatrix}
            x(t) \\ u(t) 
        \end{bmatrix} \right] \\ &= 
        \frac12 \Expect \left[\sum_{t=0}^{\bar{T}-1}\Vert 
            \Delta A_i x(t) + \Delta B_i u(t) \Vert_{\Sigma_w^{-1}}^2 
         \right],
    \end{align*}
    which concludes the proof. 
\end{proof}

\subsection{Proof of Theorem~\ref{th:SampleComplexLowerGeneral}}\label{secAppendix:SampleComplexLowerGeneral}
The first statement of the theorem follows directly from the derivations in Section~\ref{sec:ProblemSetup}.
To derive the optimal excitation input we leverage that by Theorem~\ref{th:inputDependentLowerBound} for any $\delta$-correct algorithm it holds that
\begin{equation} \label{eq:lowerBoundGeneral1}
    \Expect \left[\sum_{t=0}^{\bar{T}-1}\Vert 
    \Delta A_i x(t) + \Delta B_i u(t) \Vert_{\Sigma_w^{-1}}^2 \right] \ge 2\log\left(\frac{1}{2.4 \delta}\right) \quad \forall i \in [1, N],
\end{equation}
where $x(t)$ is generated by \eqref{eq:TrueSysEvo}.
Since \eqref{eq:lowerBoundGeneral1} needs to hold for all $i \in [1,N]$ is follows immediately that it suffices to consider 
\begin{equation} \label{eq:lowerBoundMin}
        \min_{i\in [1,N]} \Expect\left[\sum_{s=0}^{\bar{T}-1} \Vert 
        \Delta A_i x(t) + \Delta B_i u(t) \Vert_{\Sigma_w^{-1}}^2 \right] \ge 2\log\left(\frac{1}{2.4 \delta}\right) .
\end{equation}
Since we seek a result that holds for all possible input sequences that satisfy $\frac{1}{\bar{T}} \sum_{t=0}^{\bar{T}-1} u(t)^\top u(t) \le \gamma_u^2$ we maximize the l.h.s. of \eqref{eq:lowerBoundGeneral} over all admissible input sequences. 
This yields 
\begin{equation}\label{eq:proofIneqFull}
    \max_{U^\top U \le \gamma_u^2 T} \min_{i\in [1,N]}  \Expect\left[\sum_{s=0}^{t-1} \Vert \Delta A_i x(t) + \Delta B_i u(t) \Vert_{\Sigma_w^{-1}}^2 \right] \ge 2\log\left(\frac{1}{2.4 \delta}\right) .
\end{equation}
Now, using Lemma~\ref{lem:Rewritten} with $\rho=0$ and $M = \Sigma_w^{-1}$ we can equivalently rewrite the l.h.s. of \eqref{eq:proofIneqFull} as 
\begin{equation}\label{eq:rewritten}
        \max_{U^\top U \le \gamma_u^2 \bar{T}} \min_{i\in [1,N]} U^\top W_i(\bar{T})U +  
        \tr{S_w(\bar{T})^\top Q_{\Sigma_w^{-1}}^i(\bar{T}) S_w(\bar{T})}
\end{equation}
and $U^*$ is the corresponding optimizer.
For the last result, consider that 
\begin{align}
    \max_{U^\top U \le \gamma_u^2 \bar{T}} &\min_{i\in [1,N]} U^\top W_i(\bar{T})U +  
    \tr{S_w(\bar{T})^\top Q_{\Sigma_w^{-1}}^i(\bar{T}) S_w(\bar{T})} \\ 
    &\le \max_{U^\top U \le \gamma_u^2 \bar{T}} \min_{i\in [1,N]} U^\top W_i(\bar{T})U + \max_{j \in [1,N]} \tr{S_w(\bar{T})^\top Q_{\Sigma_w^{-1}}^j(\bar{T}) S_w(\bar{T})} \\
    &= \max_{U^\top U \le \gamma_u^2T} \min_{p\in \Delta_N} U^\top \left(\sum_{i=1}^{N}p_i W_i(\bar{T})\right) U + \max_{j \in [1,N]} \tr{S_w(\bar{T})^\top Q_{\Sigma_w^{-1}}^j(\bar{T}) S_w(\bar{T})} \\ 
    &= \min_{p \in \Delta_N}\max_{U^\top U \le \gamma_u^2T} U^\top \left(\sum_{i=1}^{N}p_i W_i(\bar{T})\right) U + \max_{j \in [1,N]} \tr{S_w(\bar{T})^\top Q_{\Sigma_w^{-1}}^j(\bar{T}) S_w(\bar{T})}.
\end{align}
Selecting $U^*$ in the direction of $v_\mathrm{max}\left(\sum_{i=1}^{N}p_i W_i(\bar{T})\right)$ yields the result.
\subsection{Proof of Corollary~\ref{co:lowerBoundGaussian}}\label{secAppendix:Proof_LowerBound_random}
By Theorem~\ref{th:inputDependentLowerBound} for any $\delta$-correct algorithm it holds that
    \begin{equation} \label{eq:lowerBoundGeneral}
        \Expect \left[\sum_{t=0}^{\bar{T}-1}\Vert 
        \Delta A_i x(t) + \Delta B_i u(t) \Vert_{\Sigma_w^{-1}}^2 \right] \ge 2\log\left(\frac{1}{2.4 \delta}\right) \quad \forall i \in [1, N],
    \end{equation}
    where $x(t)$ is generated by \eqref{eq:TrueSysEvo}.
    Applying Lemma~\ref{lem:Rewritten} with $\rho = 1$, $M=\Sigma_w^{-1}$ and $\sigma_u^2= \frac{\gamma_u^2}{n_u}$ we obtain that for any $\delta$-correct algorithm 
    \begin{equation}\label{eq:proofCoLowerGaussian1}
        \begin{aligned}
            \frac{\gamma_u^2}{n_u} \tr{W_i(\bar{T})} +  
            \tr{S_w(\bar{T})^\top Q_{\Sigma_w^{-1}}^i(\bar{T}) S_w(\bar{T})}
            \ge 2\log\left(\frac{1}{2.4 \delta}\right) \quad \forall i \in [1,N].
        \end{aligned}
    \end{equation}   
    After noting that $\frac{1}{n_u \bar{T}} \tr{W_i(\bar{T})} = \lambda_\mathrm{mean}(W_i(\bar{T}))$ the result follows immediately after realizing that it suffices to consider the minimum of the l.h.s. of~\eqref{eq:proofCoLowerGaussian1}. 
%
\section{Existence of the optimal solution to \eqref{eq:defOptU}}\label{app:solExists}
For clarity of exposition for each $i\in [1,N]$ we define 
\begin{align}
    W_i &\coloneqq R^i_{\Sigma_w^{-1}}(\bar{T}) + S_u(\bar{T})^\top Q^i_{\Sigma_w^{-1}}(\bar{T})  S_u(\bar{T})  + 2 N^i_{\Sigma_w^{-1}}(\bar{T}) S_u(\bar{T}),  \\
    c_i &\coloneqq \tr{S_w(\bar{T})^\top Q^i_{\Sigma_w^{-1}}(\bar{T}) S_w(\bar{T})},
\end{align} 
where $W_i \succeq 0$ and $c_i \ge 0$, for all $i \in [1,N]$.
To show that the optimal solution to \eqref{eq:defOptU} exits we consider the equivalent optimization problem 
\begin{subequations}
    \label{eq:ExistenceFull}
    \begin{align}
        \min_{U, \xi} \quad &-\xi \label{eq:ExistenceCost}\\ 
        \text{s.t. } \quad &U^\top U - \gamma_u^2 \bar{T} \le 0 \label{eq:ineqExistence1}\\
        & \xi - U^\top W_i U  - c_i  \le 0 \qquad \forall i \in [1, N]\label{eq:ineqExistence2},\\ 
        & -\xi \le 0 \label{eq:ineqSlack}
    \end{align}
\end{subequations}
where $\xi$ is a slack variable. 
Now we can leverage the following theorem due to Weierstrass.
\begin{theorem}
    Consider the constrained optimization problem $\min_{x \in \mathcal{X}} f(x)$. If the objective function $f$ is continuous and the feasible region $\mathcal{X}$ is closed and bounded, then there exists a global optimum. 
\end{theorem}
Clearly $f(\xi, U) = - \xi$ is continuous. Further, the set of feasible input sequences $U$ is closed and bounded by \eqref{eq:ineqExistence1}. 
The feasible region of $\xi$ is defined by the constraints \eqref{eq:ineqExistence2}. Clearly, for a given $\bar{T}$ under \eqref{eq:ineqExistence1} $\xi$ is upper bounded by a finite, non-negative value. 
Combining this with~\eqref{eq:ineqSlack} the feasible region of $\xi$ is bounded and closed, from which we can conclude that the optimal solution to \eqref{eq:ExistenceFull} and hence also of \eqref{eq:defOptU} exists.
\section{Proving \ac{PE} for different excitations}\label{secAppendix:Proofs_PE}
In this section we present all proofs related to \ac{PE} in the same order the results are presented in the main part of the paper. 
\subsection{Proof of Lemma~\ref{le:PEgaussianExcitation}}\label{secAppendix:proof_PE_Gaussian}
\begin{proof}
    Applying Lemma~\ref{lem:Rewritten} with $\rho = 1$ and $M=\Sigma_x^{-1}$ for any $i \in [1, N]$ yields
\begin{equation*}
    \begin{aligned}
        \frac{1}{\tau} \sum_{t=0}^{\tau}\Expect \left[\left\Vert\Delta A_i x(t) + \Delta B_i u(t)\right\Vert _{{\Sigma_w^{-1}}}^2\right]  &\ge \frac{\gamma_u^2}{n_u \tau} 
        \tr{W_i(\tau)}
        + \frac{1}{\tau} \tr{S_w(\tau)^\top Q_{\Sigma_w^{-1}}^i(\tau) S_w(\tau)}.
    \end{aligned}
\end{equation*}
Since we require a lower bound that hold uniformly over $i\in [1,N]$ we consider 
\begin{align}
    \min_{i\in[1,N]} & \frac{\gamma_u^2}{n_u \tau} 
    \tr{W_i(\tau)} + \frac{1}{\tau} \tr{S_w(\tau)^\top Q_{\Sigma_w^{-1}}^i(\tau) S_w(\tau)} \\
    &\ge \min_{i\in[1,N]} \frac{\gamma_u^2}{n_u \tau} 
    \tr{W_i(\tau)} + \min_{j\in[1,N]} \frac{1}{\tau} \tr{S_w(\tau)^\top Q_{\Sigma_w^{-1}}^j(\tau) S_w(\tau)} \\
    &= \min_{p\in\Delta_N} \frac{\gamma_u^2}{n_u \tau} 
    \tr{\sum_{i=1}^{N}p_i W_i(\tau)} + \min_{j\in[1,N]} \frac{1}{\tau} \tr{S_w(\tau)^\top Q_{\Sigma_w^{-1}}^j(\tau) S_w(\tau)} \\
    &=  \min_{p\in\Delta_N} \lambda_\mathrm{mean}\left(\sum_{i=1}^{N}p_i W_i(\tau)\right) \gamma_u^2+ \min_{j\in[1,N]} \frac{1}{\tau} \tr{S_w(\tau)^\top Q_{\Sigma_w^{-1}}^j(\tau) S_w(\tau)}. 
\end{align}
Comparing terms with \eqref{eq:defPE} yields the result.
\end{proof}
\subsection{Proof of Lemma~\ref{le:PEoptimalExcitation}}\label{secAppendix:proof_PE_Opt}
\begin{proof}
    Applying Lemma~\ref{lem:Rewritten} with $\rho = 0$ and  $M=\Sigma_w^{-1}$ for any $i\in[1, N]$ yields 
    \begin{equation}
        \begin{aligned} 
            \sum_{t=0}^{\tau-1}\Expect \left[\left\Vert
                \Delta A_i x(t) + \Delta B_i u(t)\right\Vert _{M}^2\right]  &= 
                U^\top W_i(\tau) U + 2 U^\top m_i\left(x(0)\right)  \\
                & \quad  + c_i\left(x(0)\right) + 
                \tr{S_w(\tau)^\top Q_M^i(\tau) S_w(\tau) }
        \end{aligned}
    \end{equation}
    for any $U\in \R^{n_u\tau}$.
    Thus, defining $\bar{c}_i = \tr{S_w(\tau)^\top Q_M^i(\tau) S_w(\tau)}
    $ by our choice of the excitation input we obtain 
    \begin{align}
        \max_{U^\top U \le \gamma_u^2\tau} \min_{i \in [1, N]} U^\top W_i(\tau) U &+ 2 U^\top m_i(x_0)  + c_i\left(x(0)\right) + \bar{c}_i \\ 
        &\ge \max_{U^\top U \le \gamma_u^2\tau} \min_{i \in [1, N]} U^\top W_i(\tau) U + \bar{c}_i\\
        &  \ge \max_{U^\top U \le \gamma_u^2\tau} \min_{i \in [1, N]} U^\top W_i(\tau) U + \min_{j\in [1,N]} \bar{c}_j, 
    \end{align}
    where the first inequality holds since $ c_i\left(x(0)\right)$ is non-negative by definition and $U^\top m_i(x_0)$ is non-negative when $U$ is selected to maximize the expression. 
    Similar to the proof of Theorem~\ref{th:SampleComplexLowerGeneral}, solving the optimization yields the result.   
\end{proof}

%
%
\subsection{Proof of Lemma~\ref{le:PEInputDesgin}}\label{secAppendix:proof_PE_Algo}
This proof is carried out for the case $k =0$ without loss of generality. 
To obtain the result we analyze the expected excitation of the optimal control input based on true system and the expected excitation of a suboptimal control input based on a faulty estimate separately.  
We start with 
\begin{equation}\label{eq:optInputDecomposition}
    \begin{aligned}
        \sum_{t=0}^{\tau-1}\Expect[\Vert \Delta A_i x(t) + \Delta B_i u_k^*(t)\Vert_{\Sigma_w^{-1}}^2] &= \Prob[\hat \theta_0 = \theta_*]\sum_{t=0}^{\tau-1} \Expect[\Vert \Delta A_i x(t) + \Delta B_i u_{\theta_*}^*(t)\Vert_{\Sigma_w^{-1}}^2]  \\ &\quad + \Prob[\hat \theta_0 \neq \theta_* ] \sum_{t=0}^{\tau-1} \Expect[\Vert \Delta A_i x(t) + \Delta B_i u_{\bar \theta}^*(t)\Vert_{\Sigma_w^{-1}}^2].
    \end{aligned}
\end{equation}
The first terms in \eqref{eq:optInputDecomposition} represents the excitation due to the optimal input sequence and hence after applying Lemma~\ref{lem:Rewritten} with $\nu = \rho_0$ and $M = \Sigma_w^{-1}$ we obtain 
\begin{equation}\label{eq:PE_correctID}
    \begin{aligned}
        \sum_{t=0}^{\tau-1}\Expect \left[\left\Vert\
            \Delta A_i x(t) + \Delta B_i
            u_{\theta_*}^*(t) \right\Vert _{\Sigma_w^{-1}}^2\right]  &\ge \tr{S_w(\tau)^\top Q_{\Sigma_w^{-1}}^i S_w(\tau)} 
            \\
            & \quad+ (1-\rho_0) (U_{\theta_*}^*)^\top W_i(\tau) U_{\theta_*}^*+ \frac{ \rho_0 \gamma_u^2}{n_u}
            \tr{W_i(\tau)}.
    \end{aligned}
\end{equation}
To tackle the second term in \eqref{eq:optInputDecomposition} we follow the same steps to obtain
\begin{equation} \label{eq:PE_falseID}
    \begin{aligned}
        \sum_{t=0}^{\tau-1}\Expect \left[\left\Vert
            \Delta A_i x(t) + \Delta B_i
             u_{\bar{\theta}}^*(t) 
            \right\Vert _{\Sigma_w^{-1}}^2\right]  &\ge  
            \tr{S_w(\tau)^\top Q_{\Sigma_w^{-1}}^i S_w(\tau)}
            \\ &\quad+ (1-\rho_0)(U_{\bar{\theta}}^*)^\top W_i(\tau) U_{\bar{\theta}}^* 
            + \frac{\rho_0 \gamma_u^2}{n_u} 
            \tr{W_i(\tau)}.
    \end{aligned}
\end{equation} 
Plugging \eqref{eq:PE_correctID} and \eqref{eq:PE_falseID} into \eqref{eq:optInputDecomposition} and collecting terms we obtain 
\begin{equation}\label{eq:PE_combined}
    \begin{aligned}
        \sum_{t=1}^{\tau-1}\Expect[\Vert &\Delta A_i x(t) + \Delta B_i u_k^*(t)\Vert_{\Sigma_w^{-1}}^2] \\ &\ge 
        \tr{S_w(\tau)^\top Q_{\Sigma_w^{-1}}^i S_w(\tau)}
        + \frac{\rho_0 \gamma_u^2}{ n_u} 
        \tr{W_i(\tau)}
        \\
        &\quad + \Prob[\hat \theta_0 = \theta_* ] (1-\rho_0) (U_{\theta_*}^*)^\top W_i(\tau) U_{\theta_*}^* 
        +\Prob[\hat \theta_0 \neq \theta_* ] (1-\rho_0) (U_{\bar{\theta}}^*)^\top W_i(\tau)U_{\bar{\theta}}^*.
    \end{aligned}
\end{equation}
The last term in \eqref{eq:PE_combined} is non-negative.  Hence, we obtain that 
\begin{equation}\label{eq:PE_combinedII}
    \begin{aligned}
        \sum_{t=1}^{\tau-1}&\Expect[\Vert \Delta A_i x(t) + \Delta B_i u_k^*(t)\Vert_{\Sigma_w^{-1}}^2] \\ &\ge  
        \tr{S_w(\tau)^\top Q_{\Sigma_w^{-1}}^i S_w(\tau)}
        + \frac{\rho_0 \gamma_u^2}{ n_u} 
        \tr{W_i(\tau)}
        + (1-p_0)(1-\rho_0) (U_{\theta_*}^*)^\top W_i(\tau) U_{\theta_*}^*,
    \end{aligned}
\end{equation}
where we used that $\Prob[\hat \theta_0 = \theta_* ]\ge 1- p_0$. 
Dividing by $\tau >0$ yields 
\begin{equation}\label{eq:PE_combinedFinal}
    \begin{aligned}
        &\frac{1}{\tau}\sum_{t=1}^{\tau-1}\Expect[\Vert \Delta A_i x(t) + \Delta B_i u_k^*(t)\Vert_{\Sigma_w^{-1}}^2] \\ 
        &\ge \frac{1}{\tau}
        \tr{S_w(\tau)^\top Q_{\Sigma_w^{-1}}^i S_w(\tau)}
        +  \frac{\rho_0 \gamma_u^2}{\tau n_u} 
        \tr{W_i(\tau)} + (1-p_0)(1-\rho_0) \frac{1}{\tau}(U_{\theta_*}^*)^\top W_i(\tau) U_{\theta_*}^*
        \\
        &\ge \min_{i\in [1,N]} \frac{1}{\tau}
        \tr{S_w(\tau)^\top Q_{\Sigma_w^{-1}}^i S_w(\tau)}
        +  \frac{\rho_0 \gamma_u^2}{\tau n_u} 
        \tr{W_i(\tau)} + (1-p_0)(1-\rho_0) \frac{1}{\tau}(U_{\theta_*}^*)^\top W_i(\tau) U_{\theta_*}^*.
    \end{aligned}
\end{equation}
Following similar steps as in the previous proofs yields the result.
\section{Proofs for sample complexity upper bounds}\label{secAppendix:Proofs_upperBound}
This section contains the proofs related to the sample complexity upper bounds and related intermediate results.
\subsection{Proof of Theorem~\ref{th:sampleComplexityUpperBound_General}}\label{app:ProofsampleComplexityUpperBound_Gaussian}
We provide the full version of Theorem~\ref{th:sampleComplexityUpperBound_General} and its proof.
\begin{theorem}\label{thMeta:sampleComplexityUpperBound_General}
    Consider the unknown system~\eqref{eq:TrueSysEvo}, set $\SSet$ as defined in~\eqref{eq:defSet}. Then Algorithm~\ref{alg:ID} yields an estimate $\hat{\theta}_T$ satisfying $\Prob\left[\hat \theta_T \neq \theta_*\right] \le \delta$ and with probability  at least $1-\delta$ terminates at the latest for the first $k$ satisfying 
    \begin{equation} \label{eq:stoppingTime}
        \tau\sum_{j=1}^{k}c_{u_j}(\tau) \gamma_u^2 + c_w(\tau) \ge 8\left(2a+\frac{1}{2\eta}\right) \log\left(\frac{N}{\delta}\right),
    \end{equation}
    where 
    \begin{equation}
    \eta \tau \le \frac{1}{512 \lambda_\mathrm{max}\left(\left(\Sigma_{\Delta_i}^{\frac{1}{2}}(\tau)\right)^\top \Sigma_w^{-1} \Sigma_{\Delta_i}^{\frac{1}{2}}(\tau) \right)} 
\end{equation}
    and $\Sigma_{\Delta_i}^\frac{1}{2}(\tau)$ is defined in \eqref{eq:defVarianceError}.
\end{theorem}
\begin{proof}
    This proof consists of two parts. First we show that if the algorithm terminates it holds that $\Prob\left[\hat \theta_T \neq \theta_*\right] \le \delta$, i.e., we show the algorithm is $\delta$-correct. 
    Then derive the upper bound on the stopping time and hence on the sample complexity.
    \par
    \textbf{Correctness:} For the remainder of this proof we use $\Prob_{\theta_i}$ and $\Expect_{\theta_i}$ to denote probability and expectation of an event under the hypothesis that $\theta_* = \theta_i$. 
    To analyze correctness of the algorithm we analyze the likelihood-ratio 
    \begin{equation}
        L_{\theta_i, \theta_j}(t) = \frac{\Prob_{\theta_i}[\mathcal{D}_t]}{\Prob_{\theta_j}[\mathcal{D}_t]} = \frac{\exp\left(-\frac{1}{2}\varepsilon_{\theta_i}(t)\right)}{\exp\left(-\frac12\varepsilon_{\theta_j}(t)\right)}.
    \end{equation}
    It can easily be shown that the likelihood-ratio is a martingale sequence (see, e.g., \cite[Example 2.18]{wainwright2019high}).
    By the termination rule the algorithm terminates under the event 
    \begin{align}
        \mathcal{E} &= \left\{\exists \hat \theta: \log(L_{ \hat \theta, \theta_i}(t)) > \log\left(\frac{N}{\delta}\right), \quad \forall \theta_i\in \SSet \setminus \{\hat\theta\}\right\} \\
        &= \left\{\exists \hat \theta: L_{\hat\theta, \theta_i}(t) > \frac{N}{\delta}, \quad \forall \theta_i\in \SSet \setminus \{\hat\theta\}\right\}.
    \end{align} 
    Thus, $\hat \theta \neq \theta_*$ requires $L_{\theta_i, \theta_0}(t) \ge \frac{N}{\delta}$ for at least one $\theta_i \neq \theta_*$ and some $t\in \mathbb{Z}_+$.    
    Thus 
    \begin{equation}
        \Prob\left[\hat \theta \neq \theta_*\right] \le \Prob_{\theta_0}\left[\bigcup_{i=1}^N \exists t \in \mathbb{Z}_+: L_{\theta_i, \theta_0}(t) \ge \frac{N}{\delta}\right]. 
    \end{equation} 
    Thus, 
    \begin{align}
        \Prob\left[\hat \theta \neq \theta_*\right] & \le \Prob_{\theta_0}\left[\bigcup_{i=1}^N \exists t \in \mathbb{Z}_+: L_{\theta_i, \theta_0}(t) \ge \frac{N}{\delta}\right] \\
         &\le \sum_{i=1}^N  \Prob_{\theta_0}\left[\exists t \in \mathbb{Z}_+: L_{\theta_i, \theta_0}(t) \ge \frac{N}{\delta}\right] \le \sum_{i=1}^{N} \frac{\delta}{N} = \delta,
    \end{align}
    where the second inequality uses a union bound, and the last inequality follows from~\cite[Theorem 3.9]{Lattimore2020}. 
    \par
    \textbf{Stopping Time:} 
    Since in each epoch, the input is \ac{PE} with coefficients $c_{u_j}(\tau)$ and $c_w(\tau)$ we can use Proposition~\ref{prop:NewSuperMartingale} to show that under the choice of $\eta$ the sequence 
    \begin{equation}
        S_i(k) = \exp\left(-\eta \left(\varepsilon_{\theta_i}(t_k) - \varepsilon_{\theta_*}(t_k) - \frac{1}{4} \sum_{j=0}^{k-1} \sum_{t = t_{j}}^{t_{j+1}-1} \Expect\left[\Vert\Delta A_i x(t) + \Delta B_i u(t)\Vert_{\Sigma_w^{-1}}^2 \vert \mathcal{F}_j \right]\right)\right)
\end{equation}
    with $S_i(0) = 1$ is a super-martingale.
    Hence, again using the maximal inequality~\cite[Theorem 3.9]{Lattimore2020} we obtain 
    \begin{equation}
        \Prob\left[\exists k: S_i(k) \ge \frac{N}{\delta} \right] \le \frac{\delta}{N} 
    \end{equation}
    Thus using union bound arguments we obtain 
    \begin{equation}\label{eq:proofSampleComplHelper1}
        \Prob\left[\bigcup_{i=1}^N \exists k: S_i(k) \ge \frac{N}{\delta}\right] \le \sum_{i=1}^{N}\Prob\left[\exists k: S_i(k) \ge \frac{N}{\delta} \right] \le \delta 
    \end{equation}
    Note that, given the definition of \ac{PE}, $S_i(k) \ge \frac{N}{\delta}$ is implied by 
    \begin{equation}
        \varepsilon_{\theta_*}(t_k) - \varepsilon_{\theta_i}(t_k) \ge  \frac{1}{\eta} \log\left(\frac{N}{\delta}\right) - \frac{\tau}{4}\sum_{j=0}^{k-1} c_{u_j}(\tau) \gamma_u^2 + c_w(\tau) .
    \end{equation}
    Thus by \eqref{eq:proofSampleComplHelper1} it holds with probability at least $1-\delta$ that 
    \begin{equation}
        \varepsilon_{\theta_*}(t_k) - \varepsilon_{\theta_i}(t_k) \le  \frac{1}{\eta} \log\left(\frac{N}{\delta}\right) - \frac{\tau}{4}\sum_{j=0}^{k-1} c_{u_j}(\tau) \gamma_u^2 + c_w(\tau).
    \end{equation}
    Thus, as soon as  
    \begin{equation} \label{eq:proofSampleComplBurnInCond}
        8\left(\frac{1}{2\eta}+ 2\right)\log\left(\frac{N}{\delta}\right) \le \tau \sum_{j=0}^{k-1} c_{u_j}(\tau) \gamma_u^2 + c_w(\tau)
    \end{equation}
    we obtain 
    \begin{align}
        \log\left(L_{\theta_i, \theta_*}(k\tau)\right) &= \frac12 \left(\varepsilon_{\theta_*}(t_k) - \varepsilon_{\theta_i}(t_k) \right) \\ &\le \frac{1}{2\eta} \log\left(\frac{N}{\delta}\right) - \frac{\tau}{8}\sum_{j=0}^{k-1} c_{u_j}(\tau) \gamma_u^2 + c_w(\tau) \\
        &\le \frac{1}{2\eta} \log\left(\frac{N}{\delta}\right) - \left(\frac{1}{2\eta}+ 2\right)\log\left(\frac{N}{\delta}\right) \\
        &= 2\log\left(\frac{N}{\delta}\right)
    \end{align}
    Thus, when \eqref{eq:proofSampleComplBurnInCond} is satisfied Algorithm~\ref{alg:ID} terminates. This happens at the latest when \eqref{eq:stoppingTime} is satisfied and yields $\hat \theta = \theta_*$ with probability at least $1-\delta$.
\end{proof}
\subsection{Intermediate results used in the proof of Theorem~\ref{th:sampleComplexityUpperBound_General}}
In the following, we present intermediate results which are used in the proof of Theorem~\ref{th:sampleComplexityUpperBound_General}.
\begin{prop}\label{prop:NewSuperMartingale}
    Consider the unknown system~\eqref{eq:TrueSysEvo}, and the set $\SSet$ as defined in~\eqref{eq:defSet}. 
    Let the data be collected, by exciting the true system with $u(t) = u_\mathrm{d}(t) + u_\mathrm{r}(t)$, where $u_d(t)$ is a deterministic sequence and $u_\mathrm{r}(t) \simiid \N(0, \Sigma_u)$.
    Divide the data into $K$ blocks of length $\tau>0$ and define $t_k = \tau k$, $k \in [0, K]$. Further, define the filtration $\mathcal{F}_{k}$ consisting of all random variables observed at the start of block $k$, i.e., $\mathcal{F}_{k} = \{x(0), u(0), \dots, u(t_k-1), x(t_k)\}$.
    Then the sequence
    \begin{equation}
        S_i(k) = \exp\left(-\eta \left(\varepsilon_{\theta_i}(t_k) - \varepsilon_{\theta_*}(t_k) - \frac{1}{4} \sum_{j=0}^{k-1} \sum_{t = t_{j}}^{t_{j+1}-1} \Expect\left[\Vert\Delta A_i x(t) + \Delta B_i u(t)\Vert_{\Sigma_w^{-1}}^2 \vert \mathcal{F}_j \right]\right)\right)
\end{equation}
with $S_i(0) = 1$ is a supermartingale\footnote{A supermartingale is a sequence $X(0), X(1), \dots, $ of integrable random variables satisfying $\Expect\left[X(k+1) \vert \mathcal{F}_k\right] \le X(k)$.} for any $\eta>0$ satisfying
\begin{equation}
    \eta \tau \le \frac{1}{512 \lambda_\mathrm{max}\left(\left(\Sigma_{\Delta_i}^{\frac{1}{2}}(\tau)\right)^\top \Sigma_w^{-1} \Sigma_{\Delta_i}^{\frac{1}{2}}(\tau) \right)} 
\end{equation}
with 
\begin{equation}\label{eq:defVarianceError}
    \Sigma_{\Delta_i}(\tau) \coloneqq \Delta A_i \left( \sum_{s=0}^{\tau-1} A_*^{\tau-s}B \Sigma_u B^\top (A^{\tau-s})^\top + A_*^{\tau -s} \Sigma_w (A_*^{\tau -s})^\top\right) \Delta A_i^\top + \Delta B_i \Sigma_u \Delta B_i ^\top.
\end{equation}
\end{prop}
\begin{proof}
        For the proof of this result, we overload notation and define 
    \begin{equation} 
        \varepsilon_{\theta_i}(t_0, t_1) \coloneqq \sum_{t = t_0}^{t_1-1} \Vert x(t+1) - A_i x(t) - B_i u(t)\Vert_{\Sigma_w^{-1}}^2.
    \end{equation}
    To show $S_i(k)$ is a supermartingale we consider 
    \begin{align}\label{eq:proofSupermartingale}
        \Expect\left[S_i(k+1) \vert \mathcal{F}_k\right] = &S_i(k) \Expect\Bigg[\exp\Bigg(-\eta \Bigg(\varepsilon_{\theta_i}(t_k, t_{k+1}) \\ &\quad- \varepsilon_{\theta_*}(t_k, t_{k+1}) -\frac1{4}\left(\sum_{t = t_{k}}^{t_{k+1}-1} \Expect\left[\Vert\Delta A_i x(t) + \Delta B_i u(t)\Vert_{\Sigma_w^{-1}}^2\right]\right) \Bigg)\Bigg)\bigg\vert \mathcal{F}_{k}\Bigg] .\notag
    \end{align}
    Thus, to show the result requires showing that 
    \begin{equation}
        \Expect\Bigg[e^{-\eta \left(\varepsilon_{\theta_i}(t_k, t_{k+1}) - \varepsilon_{\theta_*}(t_k, t_{k+1}) -\frac{\tau}{4}\left(\sum_{t = t_{k}}^{t_{k+1}-1} \Expect\left[\Vert\Delta A_i x(t) + \Delta B_i u(t)\Vert_{\Sigma_w^{-1}}^2\right]\right) \right)}\bigg\vert \mathcal{F}_{k}\Bigg] \le 1.
    \end{equation}
    To this end, we define   
    \begin{equation}\label{eq:DefXi}
        \xi_{\theta_i}(t) \coloneqq \Vert x(t+1) - A_i x(t) - B_i u(t)\Vert_{\Sigma_w^{-1}}^2.
    \end{equation} 
    Note that 
    \begin{align}
        \Expect\left[e^{-\eta \left(\varepsilon_{\theta_i}(t_k, t_{k+1}) - \varepsilon_{\theta_*}(t_k, t_{k+1})\right)}\Big \vert \mathcal{F}_k\right] &= \Expect\left[e^{-\eta \sum_{t=t_k}^{t_{k+1}-1} \xi_{\theta_i}(t) - \xi_{\theta_*}(t)}\Big\vert \mathcal{F}_k\right] \\
        &\quad \le \prod_{t=t_k}^{t_{k+1}-1} \left(\Expect\left[e^{-\eta \tau (\xi_{\theta_i}(t) - \xi_{\theta_*}(t))}\Big\vert \mathcal{F}_k\right]\right)^\frac{1}{\tau}  \\
        &\quad \le \prod_{t=t_k}^{t_{k+1}-1} \left(\Expect\left[e^{-\frac{\eta}{2} \tau \Vert  \Delta A_i x(t) + \Delta B_i u(t) \Vert_{\Sigma_w^{-1}}^2}\Big\vert \mathcal{F}_k\right]\right)^\frac{1}{\tau},
    \end{align}
    where the first inequality follows from Hölder's inequality and the last inequality follows from Lemma~\ref{lem:AuxUpperBound} given that $\eta \le \frac{1}{4}$.
    Note now that conditioned on $\mathcal{F}_k$ we have $\Delta A_i x(t_k + t) + \Delta B_i u(t_k +t) \sim \N(\mu_{\Delta_i}(t), \Sigma_{\Delta_i}(t))$, where $\Sigma_{\Delta_i}(t)$ is defined in \eqref{eq:defVarianceError} and $\mu_{\Delta_i}(t)$ depends on $x(t_k)$, the input sequence $u_\mathrm{d}(t)$, and the system matrices. Hence, conditioned on $\mathcal{F}_k$, by Proposition~\ref{prop:SubExpOfLipschitz} we have 
    \begin{equation}
    \begin{aligned}
        \frac{1}{2} \Big(\Vert\Delta A_i &x(t_k +t) + \Delta B_i u(t_k+t)\Vert_{\Sigma_w^{-1}}^2 \\ & - \Expect\Big[\Vert\Delta A_i x(t_k + t) + \Delta B_i u(t_k + t)\Vert_{\Sigma_w^{-1}}^2\Big] \Big) \sim \subExp{\nu(t)}, \forall t\in [0, \tau-1]
    \end{aligned}
    \end{equation}
    with $\nu(t) = 16 \lambda_\mathrm{max}\left((\Sigma_{\Delta_i}(t)^\frac12)^\top  \Sigma_w^{-1} \Sigma_{\Delta_i}(t)^\frac12\right)$. 
    Taking $\eta \tau \le \frac{1}{\nu(\tau)}$ we can use the definition of subexponential random variables to obtain 
    \begin{align}
        &\prod_{t=t_k}^{t_{k+1}-1} \left(\Expect\left[e^{-\frac{\eta}{2} \tau \Vert  \Delta A_i x(t) + \Delta B_i u(t) \Vert_{\Sigma_w^{-1}}^2}\Big\vert \mathcal{F}_k\right]\right)^\frac{1}{\tau}\\
        &\quad \le \left(\prod_{t=t_k}^{t_{k+1}-1} e^{\frac{\eta^2}2 \tau^2 \nu(t)^2 }e^{-\frac{\eta}{2} \tau \Expect\left[\Vert\Delta A_i x(t) + \Delta B_i u(t)\Vert_{\Sigma_w^{-1}}^2\right]}\right)^\frac{1}{\tau}\\
&\quad = \left(e^{-\frac{\eta}2\tau \left(\sum_{t = t_k}^{t_{k+1}-1} \Expect\left[\Vert\Delta A_i x(t) + \Delta B_i u(t)\Vert_{\Sigma_w^{-1}}^2\right] - \tau \eta \nu^2(t)\right)}\right)^\frac{1}{\tau} \\
&\quad = e^{-\frac{\eta}2 \left(\sum_{t = t_k}^{t_{k+1}-1} \Expect\left[\Vert\Delta A_i x(t) + \Delta B_i u(t)\Vert_{\Sigma_w^{-1}}^2\right] - \tau \eta \nu^2(t)\right)}.
    \end{align}
    Imposing 
    \begin{equation}\label{eq:condEta2}
        \tau \eta \nu^2(t) \le \frac{1}{2} \Expect\left[\Vert\Delta A_i x(t) + \Delta B_i u(t)\Vert_{\Sigma_w^{-1}}^2\right]
    \end{equation} for $t \in [t_k, t_{k+1}-1]$ and combining the above arguments results in 
    \begin{align}
    \Expect\left[e^{-\eta \left(\varepsilon_{\theta_i}(t_k, t_{k+1}) - \varepsilon_{\theta_*}(t_k, t_{k+1})\right)}\Big \vert \mathcal{F}_k \right] &\le \exp(-\frac{\eta}4 \sum_{t = t_k}^{t_{k+1}-1} \Expect\left[\Vert\Delta A_i x(t) + \Delta B_i u(t)\Vert_{\Sigma_w^{-1}}^2 \vert \mathcal{F}_k \right]) \le 1,
    \end{align}
    where the last inequality follows since $\Expect\left[\Vert\Delta A_i x(t) + \Delta B_i u(t)\Vert_{\Sigma_w^{-1}}^2 \vert \mathcal{F}_k \right]>0$, $\forall t \in [t_k, t_{k+1}-1]$.
    Note finally that 
    \begin{align}
        &\Expect\left[\Vert\Delta A_i x(t) + \Delta B_i u(t)\Vert_{\Sigma_w^{-1}}^2 \vert \mathcal{F}_k\right] 
        \\ \quad&= \Expect\left[(\Delta A_i x(t) + \Delta B_i u(t))^\top \Sigma_w^{-1}(\Delta A_i x(t) + \Delta B_i u(t)) \vert \mathcal{F}_k\right] \\
        &= \tr{\Sigma_w^{-1} \Expect\left[(\Delta A_i x(t) + \Delta B_i u(t)) (\Delta A_i x(t) + \Delta B_i u(t))^\top \vert \mathcal{F}_k \right]} \\
        &= \tr{\Sigma_w^{-1}(\Sigma_{\Delta_i}(t) + \mu_{\Delta_i}(t) \mu_{\Delta_i}(t)^\top)} \\
        &= \tr{(\Sigma_{\Delta_i}(t)^\frac{1}{2})^\top \Sigma_w^{-1} \Sigma_{\Delta_i}(t)^\frac{1}{2} + \mu_{\Delta_i}(t) ^\top \Sigma_w^{-1} \mu_{\Delta_i}(t)} \\
        &\ge \lambda_\mathrm{max}\left((\Sigma_{\Delta_i}(t)^\frac{1}{2})^\top \Sigma_w^{-1} \Sigma_{\Delta_i}(t)^\frac{1}{2}\right) = \frac{\nu(t)}{16},
    \end{align}
    hence condition~\eqref{eq:condEta2} is implied by
    \begin{equation}
        \tau \eta \nu^2(t) \le \frac{\nu(t)}{32} \stackrel{\nu(t) >0}{\Longleftrightarrow} \tau \eta \nu(t) \le \frac{1}{32} \forall t\in [0, \tau]
    \end{equation}
    or since the $\nu(t)$ only grows as $t$ increases by 
    \begin{equation}
        \tau \eta \le \frac{1}{32 \nu(\tau)} =  \frac{1}{512 \lambda_\mathrm{max}\left(\left(\Sigma_{\Delta_i}^{\frac{1}{2}}(\tau)\right)^\top \Sigma_w^{-1} \Sigma_{\Delta_i}^{\frac{1}{2}}(\tau) \right)} .
    \end{equation}
    \end{proof} 
\begin{lemma}\label{lem:AuxUpperBound}
    Define 
    \begin{equation}
        \xi_{\theta_i}(t) \coloneqq \Vert x(t+1) - A_i x(t) - B_i u(t)\Vert_{\Sigma_w^{-1}}^2
    \end{equation} 
    and let $\mathcal{F}_k$ be a filtration containing all random variables up to time $k$.
    Then, if $\eta \le \frac14$ we have 
    \begin{equation}
        \Expect\left[e^{-\eta (\xi_{\theta_i}(t) - \xi_{\theta_*}(t))}\Big \vert \mathcal{F}_k\right] \le  \Expect\left[\exp\left(-\frac{\eta}{2} \Big\Vert \Delta A_i x(t) - \Delta B_i u(t)\Big\Vert_{\Sigma_w^{-1}}^2 \right)\Big\vert \mathcal{F}_k \right] \quad \forall k = 1, 2, \dots , t.
    \end{equation}
\end{lemma} 
\begin{proof}
    By the definition of $\xi_{\theta_i}(t)$ 
    we have 
    \begin{align}
        \xi_{\theta_i}(t) &= \Vert \Delta A_i x(t) - \Delta B_i u(t) + w(t)\Vert_{\Sigma_w^{-1}}^2 \\ 
        &= \Vert \Delta A_i x(t) - \Delta B_i u(t)\Vert_{\Sigma_w^{-1}}^2 + 2 w(t)^\top \Sigma_w^{-1} (A_i x(t) - \Delta B_i u(t)) + w(t) ^\top \Sigma_w^{-1}  w(t)
    \end{align}
    and hence 
    \begin{equation}
        \xi_{\theta_i}(t) - \xi_{\theta_*}(t) = \Vert \Delta A_i x(t) - \Delta B_i u(t)\Vert_{\Sigma_w^{-1}}^2 + 2 w(t)^\top \Sigma_w^{-1} (A_i x(t) - \Delta B_i u(t)).
    \end{equation}
    Conditioning on $x(t)$, $u(t)$ the randomness in $\xi_{\theta_i}(t) - \xi_{\theta_*}(t)$ is due to $2 w(t)^\top \Sigma_w^{-1} (A_i x(t) - \Delta B_i u(t))$. 
    Realize that $w(t)^\top \Sigma_w^{-\frac12} \simiid \N(0, I_{n_x})$. Recall that the definition of the moment generating function for a random variable $w \sim \N (0,  I_{n_x})$ is given by 
    \begin{equation}
        M_w(\lambda) =  \Expect\left[e^{w^\top \lambda}\right] = \exp\left(\frac12  \lambda^\top \lambda \right).
    \end{equation} 
    Using this definition with $\lambda = 2 \eta {\Sigma_w^{-\frac12}}^\top (A_i x(t) - \Delta B_i u(t))$ we obtain 
    \begin{small}
    \begin{align*}
        \Expect&\left[\exp(-\eta (\xi_{\theta_i}(t) - \xi_{\theta_*}(t)))\vert x(t), u(t)\right] \\ 
        &= \exp\left(-\eta \Vert \Delta A_i x(t) + \Delta B_i u(t)\Vert_{\Sigma_w^{-1}}^2\right)  \Expect\left[\exp\left( -2\eta  w(t)^\top \Sigma_w^{-1} (A_i x(t) + \Delta B_i u(t)) \right)\vert x(t), u(t)\right] \\
        &= \exp\left(-\eta \Vert \Delta A_i x(t) + \Delta B_i u(t)\Vert_{\Sigma_w^{-1}}^2 + 2 \eta^2\Vert \Delta A_i x(t) + \Delta B_i u(t)\Vert_{\Sigma_w^{-1}}^2 \right) \\
        &= \exp\left(-\eta \Vert \Delta A_i x(t) + \Delta B_i u(t)\Vert_{\Sigma_w^{-1}}^2\left(1- 2 \eta \right) \right).
    \end{align*}
    \end{small}
    Taking $\eta \le \frac14$ the desired bound holds. 
\end{proof}

\begin{prop}\label{prop:SubExpOfLipschitz}
    Consider the random Gaussian vector $z\sim \N(\mu_z, \Sigma_z)$ of dimension $n_z$ and a symmetric matrix $M\in \mathbb{S}_{++}^{n_z}$. Then 
    $\frac12 \left(\Vert z \Vert_M^2 - \Expect\left[\Vert z\Vert_M^2\right]\right)$ is sub-exponential\footnote{A random variable X is said to be sub-exponential with parameter $\nu$ (denoted by $X\sim \subG{\nu}$) if $\Expect[X] = 0$ and its moment generating function satisfies $\Expect\left[e^{sX} \right] \le e^{\frac{s^2 \nu^2}{2}}$, $\forall \vert s\vert \le \frac{1}{\nu}$.} with parameter $\nu$, where 
    \begin{subequations}
        \begin{align}
            \nu = 16 \lambda_\mathrm{max}\left({\Sigma_z^{\frac{1}{2}}}^\top M \Sigma_z^{\frac{1}{2}}\right).
        \end{align}
    \end{subequations}
\end{prop}
\begin{proof}
    Define $\zeta = {\Sigma_z^{-\frac{1}{2}}}^\top \left(z - \mu_z\right)$, yielding $\Vert z \Vert_{M} = \Vert \zeta + {\Sigma_z^{-\frac{1}{2}}}^\top\mu_z\Vert_{{\Sigma_z^{\frac{1}{2}}}^\top M \Sigma_z^{\frac{1}{2}}}$ with $\zeta \sim\N(0, I_{n_z})$.
    Note that the $f(\zeta) = \frac1{\sqrt{2}}\Vert \zeta + \mu_z\Vert_{M}$ is Lipschitz continuous with Lipschitz constant $L = \sqrt{\lambda_\mathrm{max}(M)}$. 
    Thus, by \cite[Theorem 2.26]{wainwright2019high} 
    \begin{equation}
        \Vert \zeta + \mu_z\Vert_M \sim \subG{\lambda_\mathrm{max}\left({\Sigma_z^{\frac{1}{2}}}^\top M \Sigma_z^{\frac{1}{2}}\right)}.
    \end{equation}
    Applying \cite[Lemma 1.12]{rigollet2023highdimensionalstatistics} we obtain the result.
\end{proof}
\subsection{Sample complexity Upper bounds for isotropic Gaussian and oracle excitations}\label{secAppendix:Proofs_upperBound_Gaussian_Oracle}
In the following, we present the full version of Corollary~\ref{co:sampleComplexityRandomOracle} along with its proof.
\begin{corollary}
    Consider the same setup as in Theorem~\ref{th:sampleComplexityUpperBound_General}.
    If $u(t)\simiid\N(0, \frac{\gamma_u^2}{n_u}I_{n_u})$ then Algorithm~\ref{alg:ID} yields an estimate $\hat{\theta}$ satisfying $\Prob\left[\hat \theta \neq \theta_*\right] \le \delta$ and with probability at least $1-\delta$ terminates no later than when $k$ satisfies 
    \begin{align}
        k\bigg(\gamma_u^2 \min_{p\in\Delta_N}\lambda_\mathrm{mean}\Big(\sum_{i=1}^{N}p_i W_{i}(\tau)\Big)  &+ 
        \min_{i \in [1,N]} \frac{1}{\tau}\tr{S_w(\tau)^\top Q_{I_{n_x}}^{i} S_w(\tau)}\bigg)
        \\ &\qquad \qquad \qquad \qquad  \ge 8\left(\frac{2}{\tau}+\frac{1}{2\eta \tau}\right) \log\left(\frac{N}{\delta}\right).
    \end{align}
    Further, if the optimal oracle excitation input $U^*$ defined in \eqref{eq:oracleUTau} is applied the estimate $\hat{\theta}$ satisfies $\Prob[\hat \theta \neq \theta_*] \le \delta$, and with probability $1-\delta$ Algorithm~\ref{alg:ID} terminates no later than when $k$ satisfies 
    \begin{align}
        k \bigg(\gamma_u^2 \min_{p\in \Delta_N}\lambda_\mathrm{max}\Big(\sum_{i=1}^N p_i W_i(\tau)\Big) &+
        \min_{i\in [1, N]}\frac{1}{\tau} \tr{S_w(\tau)^\top Q_{I_{n_x}}^{i} S_w(\tau)}
        \bigg) \\
        & \qquad \qquad \qquad \qquad \ge 8\left(\frac{2}{\tau}+\frac{1}{2\eta\tau }\right) \log\left(\frac{N}{\delta}\right).
    \end{align}
    For both results we have 
    \begin{equation}
        \eta \tau \le \frac{1}{512 \lambda_\mathrm{max}\left(\left(\Sigma_{\Delta_i}^{\frac{1}{2}}(\tau)\right)^\top \Sigma_w^{-1} \Sigma_{\Delta_i}^{\frac{1}{2}}(\tau) \right)},
    \end{equation}  
    where $\Sigma_{\Delta_i}^\frac{1}{2}(\tau)$ is defined in \eqref{eq:defVarianceError}.
\end{corollary}
\begin{proof}
The result follows directly by using Theorem~\ref{thMeta:sampleComplexityUpperBound_General} with the \ac{PE} coefficients derived in Lemmas~\ref{le:PEgaussianExcitation} and~\ref{le:PEoptimalExcitation}.
\end{proof}
\subsection{Sample complexity upper bound for Algorithm~\ref{alg:inputDesign}}\label{appSection:UpperAlgo}
In the following, we present the full version of Theorem~\ref{th:sampleComlexityAlgo} along with its proof. Note that, the Theorem is presented in full generality, to provide theoretical insights. 
Depending on the update rule for $\rho_k$ corollaries can be derived directly.  
\begin{theorem} 
    Let Algorithm~\ref{alg:ID} be used with the excitation input derived by Algorithm~\ref{alg:inputDesign}. Then, the estimate $\hat \theta$ of Algorithm~\ref{alg:ID} satisfies $\Prob\left[\hat \theta \neq \theta_*\right] \le \delta$ and with probability  at least $1-\delta$ terminates no later than when $k$ satisfies
    \begin{equation}
        \sum_{j= 1}^{k} c_{u_j}^\mathrm{Alg}(\tau) + c_w(\tau) \ge 8\left(\frac{2}{\tau}+\frac{1}{2\eta \tau }\right)\log\left(\frac{N}{\delta}\right),
    \end{equation}
    where  
    \begin{equation}
        \eta \tau \le \frac{1}{512 \lambda_\mathrm{max}\left(\left(\Sigma_{\Delta_i}^{\frac{1}{2}}(\tau)\right)^\top \Sigma_w^{-1} \Sigma_{\Delta_i}^{\frac{1}{2}}(\tau) \right)},
    \end{equation}
    with $\Sigma_{\Delta_i}^\frac{1}{2}(\tau)$ as defined in \eqref{eq:defVarianceError} and 
    $c_{u_j}^\mathrm{Alg}(\tau)$ and $c_w(\tau)$ as defined in~\eqref{eq:PEALgo}.
\end{theorem}
\begin{proof}
    The result follows directly by using Theorem~\ref{thMeta:sampleComplexityUpperBound_General} and the \ac{PE} coefficients derived in Lemma~\ref{le:PEInputDesgin}.
\end{proof}
\section{Asymptotic optimality of Algorithm~\ref{alg:inputDesign}}  \label{app:asmptOpt}
The theorem below is the full version of Theorem~\ref{th:asymptOptMain}. 
\begin{theorem} \label{th:asymptOpt}
Consider Algorithm~\ref{alg:ID} where the excitation input is derived by Algorithm~\ref{alg:inputDesign}, with 
\begin{equation}
        \eta \tau \le \frac{1}{512 \lambda_\mathrm{max}\left(\left(\Sigma_{\Delta_i}^{\frac{1}{2}}(\tau)\right)^\top \Sigma_w^{-1} \Sigma_{\Delta_i}^{\frac{1}{2}}(\tau) \right)} 
    \end{equation}
where $\Sigma_{\Delta_i}^{\frac{1}{2}}(\tau) $ is defined in \eqref{eq:defVarianceError}.
Let $\rho_0 = c_0 \in (0, 1)$, $\lim_{k\to \infty}\rho_k = c_\infty \in [0, c_0]$ and $\rho_{k-1} \ge \rho_k \ge (1-\varepsilon) \rho_{k-1}$, where 
\begin{equation}
    \varepsilon < \frac{\left(1-  e^{-\frac{\eta \tau}{4} c_{w}(\tau)} \right)(1- \rho_{0})}{{\rho_{0}e^{-\frac{\eta \tau}{4} c_{w}(\tau)}}}.
\end{equation}
Then for any $\delta \in (0,1)$ there exists a $k_\delta$, such that with probability  at least $1-\delta$ Algorithm~\ref{alg:ID} terminates after at most $k_\delta$ epochs with an estimate $\hat \theta$ that satisfies $\Prob[\hat \theta \neq \theta_*] \le \delta$ 
and 
\begin{equation}
    \lim_{\delta \to 0} \frac{8 \left(\frac{1}{2\eta \tau} +\frac{1}{\tau} \right) \log\left(\frac{N}{\delta}\right) }{k_\delta}\le c_w(\tau) + \chi \gamma_u^2(1-c_\infty) c_u^\mathrm{opt}(\tau) ,
\end{equation}
where $\chi$ is a constant that can be chosen arbitrarily close to  $1$ as $\delta \to 0$.
\end{theorem}
\begin{proof}
    First note that by Propositions~\ref{prop:convergence} and \ref{prop:decrease} we know that for any $k_0 >0$ there exists a $\chi(k_0)>0$ such that  
    \begin{equation} \label{eq:lowerBoundPE}
        \forall k \ge k_0 : (1-\rho_k)(1-p_k) c_{u}^\mathrm{opt} + \rho_k c_u^\mathrm{rand} \ge \chi(k_0) (1-c_\infty)  c_u^\mathrm{opt},
    \end{equation}
    i.e., the excitation provided by Algorithm~\ref{alg:inputDesign} is $\chi(k_0)$-optimal. 
    Importantly, $\chi(k_0) \to 1$ as $k_0$ grows, hence we can push $\chi(k_0)$ arbitrarily close to $1$  by increasing $k_0$.
    Note that, correctness of Algorithm~\ref{alg:ID} follows immediately from the proof of Theorem~\ref{th:sampleComlexityAlgo}.
    Whenever $k \ge k_0$ we can use \eqref{eq:lowerBoundPE} to bound the termination time of Algorithm~\ref{alg:ID} in \eqref{eq:proofSampleComplBurnInCond} by 
    \begin{align}
         8 \left(\frac{1}{2\eta} +1 \right) \log\left(\frac{N}{\delta}\right) &\le \tau \left(kc_w + \sum_{j=1}^{k_0-1} c_{u_j}^\mathrm{Alg}(\tau) \gamma_u^2 + \sum_{j=k_0}^k \chi(k_0) (1-c_\infty) u_{u}^\mathrm{opt} \gamma_u^2 \right) \\
         &= \tau\left( k c_w + \gamma_u^2 \left((k- k_0)  \chi(k_0) (1-c_\infty) c_u^\mathrm{opt} + \sum_{j=1}^{k_0-1} c_{u_j}^\mathrm{Alg}(\tau) \right)\right)
    \end{align}
    Taking the limit $\delta \to 0$, we realize that $\delta \to 0$ requires $k \to \infty$ since all quantities on the right-hand side besides $k$ are fixed. As $k \to \infty$, the term $\sum_{j=1}^{k_0-1} c_{u_j}^\mathrm{Alg}(\tau)$ is dominated by the rest of the inequality, yielding the result.
\end{proof}
\subsection{Intermediate Results for the Proof of Theorem~\ref{th:asymptOpt}}
\begin{corollary}\label{co:ConvergenceEstimate}
    Consider the unknown system~\eqref{eq:TrueSysEvo}, set $\SSet$ as defined in~\eqref{eq:defSet} and Algorithm~\ref{alg:inputDesign}, where the input applied in round $k$ is \ac{PE} with $c_{u_j}(\tau)$ and $c_{w_j}(\tau)$. Let $\eta$ chosen to satisfy
    \begin{equation}
        \eta \tau \le \frac{1}{512 \lambda_\mathrm{max}\left(\left(\Sigma_{\Delta_i}^{\frac{1}{2}}(\tau)\right)^\top \Sigma_w^{-1} \Sigma_{\Delta_i}^{\frac{1}{2}}(\tau) \right)} 
    \end{equation}
    where $\Sigma_{\Delta_i}^{\frac{1}{2}}(\tau) $ is defined in \eqref{eq:defVarianceError}.
    Then the estimate $\hat \theta_k$ drawn in Line 2 of Algorithm~\ref{alg:inputDesign} satisfies 
    \begin{equation}
        \Prob\left[\hat \theta_k  \neq \theta_i\right] \le N\exp\left(-\frac{\eta \tau}{4}\sum_{j=1}^{k}c_{u_j}(\tau) \gamma_u^2 + c_{w}(\tau)\right).
    \end{equation}
\end{corollary}
\begin{proof}
    For the proof of this result, we overload notation and define 
    \begin{equation} 
        \varepsilon_{\theta_i}(t_0, t_1) \coloneqq \sum_{t = t_0}^{t_1-1} \Vert x(t+1) - A_i x(t) - B_i u(t)\Vert_{\Sigma_w^{-1}}^2 . 
    \end{equation}
    First note, that by similar arguments as in \cite[Lemma B.1]{muehlebach2025} we have 
    \begin{equation}\label{eq:ineqID}
        \Prob\left[\hat \theta_k = \theta_i\right] \le \Expect\left[e^{-\eta\big(\varepsilon_{\theta_i}(0, k\tau) - \varepsilon_{\theta_*}(0, k\tau) \big)}\right].
    \end{equation}
    By using Proposition~\ref{prop:NewSuperMartingale} we obtain 
    \begin{equation}
        \Prob\left[\hat \theta_k = \theta_i\right] \le 
 \Expect\left[e^{-\eta\big(\varepsilon_{\theta_i}(0, k\tau) - \varepsilon_{\theta_*}(0, k\tau) \big)}\right] \le \exp\left(-\frac{\eta \tau}{4}\sum_{j=1}^{k}c_{u_j}(\tau) \gamma_u^2 + c_{w}(\tau)\right).
    \end{equation}
\end{proof}
Note that the result can also be used to derive a similar result for the estimation of the true system using the softmax \ac{MLE}, similar to the results in \cite{muehlebach2025}. 

\begin{prop} \label{prop:convergence}
    Consider Algorithm~\ref{alg:inputDesign} with $\rho_0 = c_0 \in (0, 1)$ and $\lim_{k\to \infty} \rho_k = c_\infty \in [0, c_0]$. Let $\eta$ be chosen to satisfy
    \begin{equation}
        \eta \tau \le \frac{1}{512 \lambda_\mathrm{max}\left(\left(\Sigma_{\Delta_i}^{\frac{1}{2}}(\tau)\right)^\top \Sigma_w^{-1} \Sigma_{\Delta_i}^{\frac{1}{2}}(\tau) \right)} 
    \end{equation}
    where $\Sigma_{\Delta_i}^{\frac{1}{2}}(\tau) $ is defined in \eqref{eq:defVarianceError}. 
    Then \[\lim_{k \to \infty} c_{u_k}^\mathrm{Alg} \ge (1- c_\infty) c_u^\mathrm{opt}.\]
\end{prop}
\begin{proof}
This proof is carried out by contradiction. Suppose there exists some $\tilde{c} \in [0, 1)$ such that 
\begin{equation}
    \lim_{k \to \infty} c_{u_k}^\mathrm{Alg} = \tilde c (1- c_\infty) c_u^\mathrm{opt} < (1- c_\infty) c_u^\mathrm{opt}.
\end{equation}
plugging in the definition of $c_{u_k}^\mathrm{Alg}$ we see this requires
\begin{equation} \label{eq:convergenceinterm1}
    \lim_{k\to \infty} (1- p_k)  = \tilde{c}   \Leftrightarrow \lim_{k\to \infty} p_k = 1-\tilde{c}
\end{equation}
Invoking Corollary~\ref{co:ConvergenceEstimate} we that 
\begin{equation}\label{eq:defpk}
    p_k \le N \exp\left(-\frac{\eta \tau}{4} \left(k c_w(\tau) + \sum_{j=1}^k c^\mathrm{Alg}_{u_j}(\tau) \gamma_u^2 \right)\right).
\end{equation}
Hence, \eqref{eq:convergenceinterm1} requires 
\begin{equation}\label{eq:convergenceinterm2}
    \lim_{k\to \infty} kc_w +  \sum_{j=1}^k c^\mathrm{Alg}_{u_j}(\tau) \gamma_u^2  = \frac{4}{\tau \eta} \log\left(\frac{1-\tilde c}{N}\right).
\end{equation}
Since $c_w \ge 0$ we require $c_w = 0$. This is unlikely but for the remainder, we consider this case to be true. In any case, we additionally require 
\begin{equation}\label{eq:convergenceinterm3}
    \lim_{k\to \infty} c_{u_k}^\mathrm{Alg} = 0
\end{equation}
since if this is not true, the sum in \eqref{eq:convergenceinterm2} does not converge as $k \to \infty$. 
Going back to the definition of $c_{u_k}^\mathrm{Alg}$ we see that \eqref{eq:convergenceinterm2} can only be satisfied if the conditions 
\begin{itemize}
    \item $\lim_{k\to \infty} \rho_k = 0$
    \item $\lim_{k\to \infty} (1-\rho_k)(1-p_k) = 0$
\end{itemize}
are both met. Clearly, if $c_\infty \neq 0$, the first condition is not met, yielding the result for all $c_\infty \in (0, c_0]$. 
For the case where $c_\infty = 0$ the second condition requires that $\lim_{k\to \infty} p_k = 1$. Going back to \eqref{eq:defpk}, we see that for this to happen, we require  
\begin{equation}
    \lim_{k\to \infty} \sum_{j=1}^k c_{u_j}^\mathrm{Alg} \gamma_u^2 = 0. 
\end{equation}
Since we selected $\rho_0 = c_0>0$ we obtain $c_{u_0}^\mathrm{Alg} \ge c_0 c_u^\mathrm{rand} >0$ which contradicts the condition. Hence, there cannot exist a $\tilde c \in [0, 1)$  for which 
\begin{equation}
    \lim_{k \to \infty} c_{u_k}^\mathrm{Alg} = \tilde c (1- c_\infty) c_u^\mathrm{opt} < (1- c_\infty) c_u^\mathrm{opt}.
\end{equation}
which completes the proof.
\end{proof}
\begin{prop} \label{prop:decrease}
Consider Algorithm~\ref{alg:inputDesign} with $\rho_0 = c_0 \in (0, 1)$ and $\rho_{k-1} \ge \rho_k \ge (1-\varepsilon) \rho_{k-1}$. 
Let $\eta$ be chosen to satisfy
    \begin{equation}
        \eta \tau \le \frac{1}{512 \lambda_\mathrm{max}\left(\left(\Sigma_{\Delta_i}^{\frac{1}{2}}(\tau)\right)^\top \Sigma_w^{-1} \Sigma_{\Delta_i}^{\frac{1}{2}}(\tau) \right)} 
    \end{equation}
where $\Sigma_{\Delta_i}^{\frac{1}{2}}(\tau) $ is defined in \eqref{eq:defVarianceError}.
If $\varepsilon < \frac{\left(1-  e^{-\frac{\eta \tau}{4} c_{w}(\tau)} \right)(1- \rho_{0})}{{\rho_{0}e^{-\frac{\eta \tau}{4} c_{w}(\tau)}}}$, then $c_{u_k}^\mathrm{Alg} \ge c_{u_{k-1}}^\mathrm{Alg}$ for all $k \ge 1$.
\end{prop}
\begin{proof}
    Plugging in the definitions of the PE coefficients, the statement of the theorem can be rewritten as 
    \begin{align}
        (1-\rho_k) (1-p_k) c_{u}^\mathrm{opt} + \rho_k c_{u}^\mathrm{rand}  >  (1-\rho_{k-1}) (1-p_{k-1}) c_{u}^\mathrm{opt} + \rho_{k-1} c_{u}^\mathrm{rand}.   
    \end{align}
    Rearanging terms, we obtain 
    \begin{alignat}{2}
        & c_{u}^\mathrm{opt}\Big( (1-\rho_k) (1-p_k) - (1-\rho_{k-1}) (1-p_{k-1}) \Big) &&> c_{u}^\mathrm{rand} (\rho_{k-1} - \rho_k) \\
        \Leftrightarrow \quad & c_{u}^\mathrm{opt}\Big(1 - \rho_k - p_k + p_k\rho_k - (1 - \rho_{k-1} - p_{k-1} + p_{k-1}\rho_{k-1} )\Big) &&> c_{u}^\mathrm{rand} (\rho_{k-1} - \rho_k)  \\
        \Leftrightarrow \quad & c_{u}^\mathrm{opt} (\rho_{k-1} - \rho_k + p_{k-1} - p_k + p_k \rho_k - p_{k-1} \rho_{k-1}) &&>  c_{u}^\mathrm{rand} (\rho_{k-1} - \rho_k) \label{eq:AsymptOptimalInterm1}  
    \end{alignat}
    Clearly, since $\rho_k$ is a non-increasing sequence, \eqref{eq:AsymptOptimalInterm1} is satisfied whenever 
    \begin{equation}
        p_{k-1} - p_k + p_k \rho_k - p_{k-1} \rho_{k-1}>0 \label{eq:AsymptOptimalInterm2}
    \end{equation}
    Note that, by Corollary~\ref{co:ConvergenceEstimate} we have $p_k = N\exp\left(-\frac{\eta \tau}{4}\sum_{j=1}^{k}c_{u_j}(\tau) \gamma_u^2 + c_{w}(\tau)\right)$ and hence 
    \begin{equation}
        p_k = p_{k-1} e^{-\frac{\eta \tau}{4} (c_{u_k}(\tau) \gamma_u^2 + c_{w}(\tau))}.
    \end{equation}
    Using this, and the fact that $\rho_k \ge \rho_{k-1}(1-\varepsilon)$ we have that  \eqref{eq:AsymptOptimalInterm2} is implied by
    \begin{align}
        p_{k-1}\left( 1 - e^{-\frac{\eta \tau}{4} (c_{u_k}(\tau) \gamma_u^2 + c_{w}(\tau))}+ \rho_{k-1} (1-\varepsilon) e^{-\frac{\eta \tau}{4} (c_{u_k}(\tau) \gamma_u^2 + c_{w}(\tau))}   - \rho_{k-1}  \right) > 0,
    \end{align}
    which can only be satisfied if $ 1 - e^{-\frac{\eta \tau}{4} (c_{u_k}(\tau) \gamma_u^2 + c_{w}(\tau))}+ \rho_{k-1} (1-\varepsilon) e^{-\frac{\eta \tau}{4} (c_{u_k}(\tau) \gamma_u^2 + c_{w}(\tau))}  - \rho_{k-1} >0$. 
    Using standard algebraic manipulations, we can rewrite this as
    \begin{equation}
        \left(1-  e^{-\frac{\eta \tau}{4} (c_{u_k}(\tau) \gamma_u^2 + c_{w}(\tau))} \right)(1- \rho_{k-1}) >  \rho_{k-1} \varepsilon  e^{-\frac{\eta \tau}{4} (c_{u_k}(\tau) \gamma_u^2 + c_{w}(\tau))}.
    \end{equation}
    Solving this for $\varepsilon$ we obtain 
    \begin{equation}
        \varepsilon < \frac{\left(1-  e^{-\frac{\eta \tau}{4} (c_{u_k}(\tau) \gamma_u^2 + c_{w}(\tau))} \right)(1- \rho_{k-1})}{\rho_{k-1}e^{-\frac{\eta \tau}{4} (c_{u_k}(\tau) \gamma_u^2 + c_{w}(\tau))}}.
    \end{equation}
    Since we require an upper bound for all $k\ge 1$, we use the (loose) lower bound to the right-hand side 
    \begin{equation}
         \frac{\left(1-  e^{-\frac{\eta \tau}{4} (c_{u_k}(\tau) \gamma_u^2 + c_{w}(\tau))} \right)(1- \rho_{k-1})}{\rho_{k-1}e^{-\frac{\eta \tau}{4} (c_{u_k}(\tau) \gamma_u^2 + c_{w}(\tau))}} >  \frac{\left(1-  e^{-\frac{\eta \tau}{4} c_{w}(\tau)} \right)(1- \rho_{0})}{{\rho_{0}e^{-\frac{\eta \tau}{4} c_{w}(\tau)}}}, 
    \end{equation}
    to obtain that if $\varepsilon < \frac{\left(1-  e^{-\frac{\eta \tau}{4} c_{w}(\tau)} \right)(1- \rho_{0})}{{\rho_{0}e^{-\frac{\eta \tau}{4} c_{w}(\tau)}}}$ we have $c_{u_{k}}^\mathrm{Alg} > c_{u_{k-1}}^\mathrm{Alg}$ for all $k \ge 1$.
\end{proof}
Note that by Proposition~\ref{prop:decrease} as long as $\rho_0 \neq 1$ we can show that the sequence $\{c_{u_{k}}^\mathrm{Alg}\}_{k\ge 0}$ increases at every timestep. However, even if we initially collect some data using random excitations (i.e., we select $\rho_0 = 1$), we can still guarantee that after a potential initial decay $\{c_{u_{k}}^\mathrm{Alg}\}_{k\ge 1}$ increases at every timestep. 

\section{Auxiliary results} \label{secAppendix:AdditionalResults}
The following result is a consequence of the superposition principle in \ac{LTI} systems and leverages that all random quantities are independent. It is a core ingredient in the proofs of the sample complexity lower and upper bounds.
\begin{lemma} \label{lem:Rewritten}
    Consider the system \eqref{eq:TrueSysEvo} with $x(0) \in \R^{n_x}$. Let $u(t) = (1-\rho) u_\mathrm{d}(t) + \rho u_\mathrm{p}(t)$, with $\rho\in[0,1]$,  $u_\mathrm{p}(t) \simiid \N(0, \sigma_u^2 I_{n_u})$ and a deterministic control input $u_\mathrm{d}(t)$.
    Then for any fixed matrix $M\in \mathbb{S}_{++}^{n_x}$, any $i \in [1,N]$, and any $\tau\ge0$ it holds that 
        \begin{equation}
            \label{eq:Lemma_appendix_statement}
            \begin{aligned} 
                \sum_{t=0}^{\tau}\Expect \left[\left\Vert
                    \Delta A_i x(t) + \Delta B_i u(t)\right\Vert _{M}^2\right]  &= (1-\rho)^2
                    U_\mathrm{d}^\top W_i(\tau) U_\mathrm{d} + 2(1-\rho) U_\mathrm{d}^\top m_i(x_0) + c_i(x(0)) \\
                    & \quad + \sigma_u^2 \rho^2 
                    \tr{W_i(\tau)} + \tr{S_w(\tau)^\top Q_M^i(\tau) S_w(\tau)},
            \end{aligned}
        \end{equation}
        where
        \begin{align*}
            W_i(\tau) &\coloneqq R_M^i(\tau)+ S_u(\tau)^\top Q_M^i(\tau) S_u(\tau) + N_M^i(\tau)S_u(\tau) + (N_M^i(\tau)S_u(\tau))^\top \\
            m_i(x(0)) &\coloneqq \left( S_u(T)^\top Q_M^i(\tau) S_w(\tau) X(0) +   N_M^i(\tau) S_w(\tau) X(0) \right) \\
            c_i(x(0)) &\coloneqq \sum_{t=0}^{\tau} \left(\Delta A_i \sum_{s=0}^{t-1} A_*^s x(0)\right)^\top M \left(\Delta A_i \sum_{s=0}^{t-1} A_*^s x(0)\right)\\
            X(0)^\top &\coloneqq \begin{bmatrix}
                (\Sigma_w^{-\frac12} x)^\top & 0 & \dots & 0  
            \end{bmatrix}.
        \end{align*}
        Furthermore, 
    \begin{itemize}
        \item If $x(0) = 0$ then it holds for all $\rho \in [0, 1]$ that 
        \begin{equation}
            \label{eq:lemma_statement1}
            \begin{aligned}
                \sum_{t=0}^{\tau}\Expect \left[\left\Vert
                    \Delta A_i x(t) + \Delta B_i u(t) \right\Vert _{M}^2\right]  & \ge (1-\rho)^2 U_\mathrm{d}^\top W_i(\tau) U_\mathrm{d}  +\sigma_u^2 \rho^2 
                    \tr{W_i(\tau)} \\ &\quad+ \tr{S_w(\tau)^\top Q_M^i(\tau) S_w(\tau)}
            \end{aligned}
        \end{equation}
        \item If $\rho = 1$ then it holds for all $x(0) \in \R^{n_x}$ that 
        \begin{equation}
            \label{eq:lemma_statement2}
            \begin{aligned}
                \sum_{t=0}^{\tau}\Expect \left[\left\Vert
                    \Delta A_i x(t) + \Delta B_i u(t) \right\Vert _{M}^2\right]  & \ge \sigma_u^2 \tr{W_i(\tau)} + \tr{S_w(\tau)^\top Q_M^i(\tau) S_w(\tau)}
            \end{aligned}
        \end{equation}
    \end{itemize}
\end{lemma}
\begin{proof}
By recursively plugging in the $x(t) = A_* x(t-1) + Bu(t-1) + w(t-1)$ we obtain  
\begin{align}
    \Delta A_i  x(t) + \Delta B_i u(t) &= 
    \Delta A_i  x(t) + (1-\rho) \Delta B_i u_\mathrm{d}(t) +  \rho\Delta B_i u_\mathrm{p}(t)
     \\ 
    &= \underbrace{ \Delta A_i\left(\sum_{s= 0}^{t-1} (1-\rho) A_*^{s-t-1}B_* u_\mathrm{d}(s)  +  A_*^s x(0)\right)  +  (1-\rho) \Delta B_i u_\mathrm{d}(t)}_{\text{deterministic}} \label{eq:proofHelper1}\\
    &\quad+ \underbrace{\rho \Delta A_i \sum_{s= 0}^{t-1}A_*^{s-t-1}B_*  u_\mathrm{p}(t) + \rho \Delta B_i  u_\mathrm{p}(t)}_{\text{rand. vec., dep. on }u_\mathrm{p}} + \underbrace{\Delta A_i  \sum_{s= 0}^{t-1}A_*^{s-t-1} w(t)}_{\text{rand. vec., dep. on }w}. \notag
\end{align}
Using that all, random vectors are independent and zero mean, and thus, cross terms are zero, it follows by plugging in \eqref{eq:proofHelper1} that
\begin{equation}\label{eq:lemma_app_1}
\begin{aligned}
    \sum_{t=0}^{\tau}\Expect &[\Vert
        \Delta A_i x(t) + \Delta B_i u(t)
        \Vert _{M}^2]  
        \\ 
        &= \sum_{t=0}^{\tau}\Expect \left[
            \left(\Delta A_i x(t) + \Delta B_i u(t)\right)^\top M \left(\Delta A_i x(t) + \Delta B_i u(t)\right) 
            \right] \\
            &=   \sum_{t=0}^{\tau} \Bigg(\left\Vert
            \Delta A_i\left((1-\rho)\sum_{s= 0}^{t-1}A_*^{s-t-1}B_* u_\mathrm{d}(s)  +  A_*^s x(0)\right)  +  (1-\rho)\Delta B_i u_\mathrm{d}(t)\right\Vert _{M}^2 \\
        & \qquad \qquad\quad+\rho^2\Expect\left[ \left\Vert  \Delta B_i  u_\mathrm{p}(t) + \Delta A_i \sum_{s= 0}^{t-1}A_*^{s-t-1}B_*  u_\mathrm{p}(s)\right\Vert _M^2 \right] \\
        &\qquad \qquad \quad+ \Expect\left[\left\Vert \Delta A_i  \sum_{s= 0}^{t-1}A_*^{s-t-1} w(s) \right\Vert _M^2 \right] \Bigg).
\end{aligned}
\end{equation}
We carry on by rewriting the first term in \eqref{eq:lemma_app_1}
\begin{equation}\label{eq:lemma_app_2}
    \begin{aligned}
        \sum_{t=0}^{\tau} \Big\Vert
            &\Delta A_i\left((1-\rho)\sum_{s= 0}^{t-1}A_*^{s-t-1}B_* u_\mathrm{d}(s)  +  A_*^s x(0)\right)  +  (1-\rho)\Delta B_i u_\mathrm{d}(t)\Big\Vert _{M}^2 \\ 
            &= (1-\rho)^2 U_\mathrm{d} ^\top W_i(\tau) U_\mathrm{d} + \underbrace{\sum_{t=0}^{\tau} \left(\Delta A_i \sum_{s=0}^{t-1} A_*^s x(0)\right)^\top M \left(\Delta A_i \sum_{s=0}^{t-1} A_*^s x(0)\right)}_{\coloneqq c_i(x(0)) \geq 0}  \\
            &\quad + 2(1-\rho) U_\mathrm{d}^\top \underbrace{\left( S_u(\tau)^\top Q_M^i(\tau) S_w(\tau) X(0) +   N_M^i(\tau) S_w(\tau) X(0) \right)}_{\coloneqq m_i(x_0)},
    \end{aligned}
\end{equation}
where $X(0)^\top = \begin{bmatrix}
    (\Sigma_w^{-\frac12} x)^\top & 0 & \dots & 0  
\end{bmatrix}$.
Let us first consider $x(0) = 0$. 
Then we can use the matrices defined in~\eqref{eq:defToeplitz} and \eqref{eq:defMatrices} and plug them into~\eqref{eq:lemma_app_2} to obtain
\begin{equation}
    \begin{aligned}
        \sum_{t=0}^{\tau}\Expect \left[\left\Vert 
            \Delta A_i x(t) + \Delta B_i u(t) \right\Vert _{M}^2\right]  &= (1-\rho)^2
            U_\mathrm{d}^\top W_i(\tau) U_\mathrm{d} + 2(1-\rho) U_\mathrm{d}^\top m_i(x_0) + c_i(x(0)) \\
        & \quad + \rho^2
        \Expect\left[U_\mathrm{p} W_i(\tau) U_\mathrm{p}\right] + \Expect\left[W^\top S_w(\tau)^\top Q_M^i(\tau) S_w(\tau) W\right]
    \end{aligned}
\end{equation}
Finally, we use $\Expect[x^\top M x] = \tr{M \Expect\left[ x x^\top \right]}$ and $\tr{A} = \sum_{j=1}^{n} \lambda_j(A)$, where $\lambda_j(A)$ is the $j$-th eigenvalue of the matrix $A\in \R^{n\times n}$ to obtain 
\begin{equation}
    \label{eq:Lemma_appendix_proof_final}
    \begin{aligned} 
        \sum_{t=0}^{\tau}\Expect \left[\left\Vert
            \Delta A_i x(t) + \Delta B_i u(t)\right\Vert _{M}^2\right]  &= (1-\rho)^2
            U_\mathrm{d}^\top W_i(\tau) U_\mathrm{d} + 2(1-\rho) U_\mathrm{d}^\top m_i(x_0) + c_i(x(0)) \\
            & \quad + \sigma_u^2 \rho^2 
            \tr{W_i(\tau)} + \tr{S_w(\tau)^\top Q_M^i(\tau) S_w(\tau)}
    \end{aligned}
\end{equation}
which proves the first statement of the Lemma. The second statement follows by the fact that $m_i(0) = 0$ and $c_i(0) = 0$. The last statement follows, plugging in $\rho = 1$ and using $c_i(x(0)) \ge 0$, $\forall x(0)\in \R^{n_x}$. 
\end{proof}

\section{Numerical evaluation of Example~\ref{ex:Example1}}\label{appSection:numerics}
In this section, we consider a variant of the setup considered in Example~\ref{ex:Example1}. 
To be precise, consider $\tilde{\SSet}=\{(\tilde{A}_*, \tilde{B}_*), (\tilde{A}_1, \tilde{B}_*)\}$ with 
\begin{alignat*}{2}
        \tilde{A}_* &= \begin{bmatrix}
            0 & 0.1 & 0 & 0 & 0 & 0\\ 
            0 & 0 & 0 & 0 & 0 &0 \\ 
            0 & 0 & 0.9 & 0 & 0 & 0\\ 
            0 & 0 & 0 & 0.9 & 0 & 0 \\
            0 & 0 & 0 & 0 & 0.9 & 0 \\
            0 & 0 & 0 & 0 & 0 & 0.9
        \end{bmatrix}, \qquad 
        \tilde{B}_*&&= \begin{bmatrix}
            0 & 0 &0 & 0 & 0\\ 1 & 0 & 0 & 0 & 0 \\ 0 & 1 & 0 & 0& 0\\
             0& 0 &1 & 0 & 0 \\ 0 & 0 & 0 & 1& 0 \\ 0 & 0 & 0 & 0& 1
        \end{bmatrix},
        \\
        \tilde{A}_1 &= \begin{bmatrix}
            0 & 0.2 & 0 & 0 & 0 & 0\\ 
            0 & 0 & 0 & 0 & 0 &0 \\ 
            0 & 0 & 0.9 & 0 & 0 & 0\\ 
            0 & 0 & 0 & 0.9 & 0 & 0 \\
            0 & 0 & 0 & 0 & 0.9 & 0 \\
            0 & 0 & 0 & 0 & 0 & 0.9
        \end{bmatrix},
        \qquad
        \tilde{B}_1 &&= \tilde{B}_*.
    \end{alignat*}
Further, we fix $\gamma_u=1$ and consider the process noise $w(t)\simiid\N(0, \sigma_w^2 I_6)$, with $\sigma_w = 0.1$. Following the steps outlined in Example~\ref{ex:Example1} we obtain $\lambda_\mathrm{max}(\tilde{W}_1(t)) = 5 \lambda_\mathrm{mean}(W_1(t))$.
Since we initialize Algorithm~\ref{alg:inputDesign} with equal weights for both systems we immediately obtain $\delta_0 = \frac12$. Thus, by the results in Section~\ref{sec:InputDesignAlgo} we expect \ac{CE} to perform well. 
To highlight the effectiveness of Algorithm~\ref{alg:inputDesign} we compare 
$\Prob[\hat\theta_k = \theta_*]$ for three different setups, where we select $\eta = 0.1$: 
\begin{itemize}
    \item Algorithm~\ref{alg:inputDesign} with $\rho_k \equiv 0$, i.e., greedily using \ac{CE} with the current estimate 
    \item Algorithm~\ref{alg:inputDesign} with $\rho_k \equiv 1$, i.e., isotropic Gaussian excitations 
    \item  Algorithm~\ref{alg:inputDesign} with $\rho_k \equiv 0$ and oracle knowledge, i.e., computing the optimal excitation with $\theta_*$
\end{itemize}
We select $\tau = 10$ and run $20$ Monte Carlo simulations.
The mean and $\sigma$-bands of the likelihood of $\theta_*$ given the data are displayed in Figure~\ref{fig:numericsToyExample}. 
It is immediately apparent that oracle knowledge does not provide any advantage in this setup, since the current estimate $\hat \theta_k$ does not influence the solution of~\eqref{eq:runningOptU} due to the structure of the problem. 
Further, as expected, the experiment design algorithm significantly outperforms the random excitations by allocating all input energy to the first input. 
\begin{figure}[h]
    \centering
    \resizebox{0.5\linewidth}{!}{
%
%
\definecolor{mycolor1}{rgb}{0.49400,0.18400,0.55600}%
\definecolor{mycolor2}{rgb}{0.30100,0.74500,0.93300}%
\definecolor{mycolor3}{rgb}{0.63500,0.07800,0.18400}%
\begin{tikzpicture}

\begin{axis}[%
width=5.5in,
height=2.75in,
at={(0in,0in)},
scale only axis,
xmin=0,
xmax=40,
xlabel style={font=\color{white!15!black}},
xlabel={Epoch},
ymin=0.45,
ymax=1.01,
ytick={0.5, 0.6, 0.7, 0.8, 0.9, 1},
ylabel style={font=\color{white!15!black}},
ylabel={Likelihood of $\theta_*$},
axis background/.style={fill=white},
axis x line*=bottom,
axis y line*=left,
xmajorgrids,
ymajorgrids,
legend style={at={(0.97,0.03)}, anchor=south east, legend cell align=left, align=left, draw=white!15!black}
]
\addplot [color=mycolor1, line width=2.0pt]
  table[row sep=crcr]{%
0	0.5\\
1	0.637652215729943\\
2	0.740053224882772\\
3	0.844168050351662\\
4	0.889714131179595\\
5	0.926139542450592\\
6	0.950578587094164\\
7	0.965511192953144\\
8	0.975787941258588\\
9	0.98230857474357\\
10	0.985980145842724\\
11	0.991326685943406\\
12	0.994666097179504\\
15	0.997821739606643\\
16	0.998686343883932\\
18	0.999531166992455\\
22	0.999925545266905\\
31	0.999996918299225\\
40	0.999999912678156\\
};
\addlegendentry{Alg.~\ref{alg:inputDesign} with $\rho_k \equiv 0$}

\addplot[area legend, draw=none, fill=mycolor1, fill opacity=0.1, forget plot]
table[row sep=crcr] {%
x	y\\
0	0.5\\
1	0.557887206350693\\
2	0.655307323867259\\
3	0.76777306324897\\
4	0.82085399430067\\
5	0.871022914387103\\
6	0.905322706861659\\
7	0.935469922271964\\
8	0.953404196006001\\
9	0.962627055842085\\
10	0.962541749420211\\
11	0.978704418166612\\
12	0.987393986092434\\
13	0.988704555892101\\
14	0.992195783235366\\
15	0.994448014477039\\
16	0.996618363601129\\
17	0.997949917717063\\
18	0.99884298989776\\
19	0.999233441305618\\
20	0.999441920133639\\
21	0.999680573353399\\
22	0.999813340897458\\
23	0.999899128435806\\
24	0.999932469593624\\
25	0.999965632061924\\
26	0.999981643443366\\
27	0.999982709620267\\
28	0.999984474534847\\
29	0.999989362544116\\
30	0.999989006483601\\
31	0.99998687436873\\
32	0.999990459346913\\
33	0.999992742681837\\
34	0.9999957236166\\
35	0.999997237673828\\
36	0.999998521806096\\
37	0.999998739888249\\
38	0.999998806827474\\
39	0.999999341087552\\
40	0.999999576184953\\
40	1.00000024917136\\
39	1.00000039129678\\
38	1.00000072590179\\
37	1.00000075247961\\
36	1.00000084941655\\
35	1.00000157595433\\
34	1.00000238745451\\
33	1.00000407579825\\
32	1.00000516989378\\
31	1.00000696222973\\
30	1.00000564336976\\
29	1.00000458162251\\
28	1.00000660521259\\
27	1.00000590613687\\
26	1.0000035949403\\
25	1.00000844381527\\
24	1.00001761405129\\
23	1.00002258541345\\
22	1.00003774963635\\
21	1.00007377662371\\
20	1.00012450611471\\
19	1.00013605146974\\
18	1.00021934408716\\
17	1.00040551699381\\
16	1.00075432416674\\
15	1.00119546473625\\
14	1.00179736628292\\
13	1.00286960157176\\
12	1.00193820826657\\
11	1.0039489537202\\
10	1.00941854226524\\
9	1.00199009364505\\
8	0.998171686511174\\
7	0.995552463634331\\
6	0.995834467326674\\
5	0.981256170514075\\
4	0.958574268058521\\
3	0.920563037454352\\
2	0.82479912589828\\
1	0.717417225109194\\
0	0.5\\
}--cycle;
\addplot [color=mycolor2, dashed, line width=2.0pt]
  table[row sep=crcr]{%
0	0.5\\
1	0.637652215729943\\
2	0.740053224882772\\
3	0.844168050351662\\
4	0.889714131179595\\
5	0.926139542450592\\
6	0.950578587094164\\
7	0.965511192953144\\
8	0.975787941258588\\
9	0.98230857474357\\
10	0.985980145842724\\
11	0.991326685943406\\
12	0.994666097179504\\
15	0.997821739606643\\
16	0.998686343883932\\
18	0.999531166992455\\
22	0.999925545266905\\
31	0.999996918299225\\
40	0.999999912678156\\
};
\addlegendentry{Alg.~\ref{alg:inputDesign} with oracle knowledge}

\addplot[area legend, draw=none, fill=mycolor2, fill opacity=0.1, forget plot]
table[row sep=crcr] {%
x	y\\
0	0.5\\
1	0.557887206350693\\
2	0.655307323867259\\
3	0.76777306324897\\
4	0.82085399430067\\
5	0.871022914387103\\
6	0.905322706861659\\
7	0.935469922271964\\
8	0.953404196006001\\
9	0.962627055842085\\
10	0.962541749420211\\
11	0.978704418166612\\
12	0.987393986092434\\
13	0.988704555892101\\
14	0.992195783235366\\
15	0.994448014477039\\
16	0.996618363601129\\
17	0.997949917717063\\
18	0.99884298989776\\
19	0.999233441305618\\
20	0.999441920133639\\
21	0.999680573353399\\
22	0.999813340897458\\
23	0.999899128435806\\
24	0.999932469593624\\
25	0.999965632061924\\
26	0.999981643443366\\
27	0.999982709620267\\
28	0.999984474534847\\
29	0.999989362544116\\
30	0.999989006483601\\
31	0.99998687436873\\
32	0.999990459346913\\
33	0.999992742681837\\
34	0.9999957236166\\
35	0.999997237673828\\
36	0.999998521806096\\
37	0.999998739888249\\
38	0.999998806827474\\
39	0.999999341087552\\
40	0.999999576184953\\
40	1.00000024917136\\
39	1.00000039129678\\
38	1.00000072590179\\
37	1.00000075247961\\
36	1.00000084941655\\
35	1.00000157595433\\
34	1.00000238745451\\
33	1.00000407579825\\
32	1.00000516989378\\
31	1.00000696222973\\
30	1.00000564336976\\
29	1.00000458162251\\
28	1.00000660521259\\
27	1.00000590613687\\
26	1.0000035949403\\
25	1.00000844381527\\
24	1.00001761405129\\
23	1.00002258541345\\
22	1.00003774963635\\
21	1.00007377662371\\
20	1.00012450611471\\
19	1.00013605146974\\
18	1.00021934408716\\
17	1.00040551699381\\
16	1.00075432416674\\
15	1.00119546473625\\
14	1.00179736628292\\
13	1.00286960157176\\
12	1.00193820826657\\
11	1.0039489537202\\
10	1.00941854226524\\
9	1.00199009364505\\
8	0.998171686511174\\
7	0.995552463634331\\
6	0.995834467326674\\
5	0.981256170514075\\
4	0.958574268058521\\
3	0.920563037454352\\
2	0.82479912589828\\
1	0.717417225109194\\
0	0.5\\
}--cycle;
\addplot [color=mycolor3, line width=2.0pt]
  table[row sep=crcr]{%
0	0.5\\
1	0.540805481779408\\
2	0.579293911316434\\
3	0.606709451956448\\
4	0.632639395130909\\
5	0.657389448441478\\
6	0.691998262090102\\
7	0.716858023699196\\
8	0.742654433567026\\
9	0.756855370153204\\
10	0.773718029713955\\
11	0.782828921758586\\
12	0.802648199950852\\
13	0.809043333515632\\
14	0.816584976823073\\
15	0.82349230450744\\
16	0.83753111855664\\
17	0.846909777573913\\
18	0.859498407190614\\
20	0.876748534386152\\
21	0.889099012255606\\
22	0.896826504464215\\
23	0.905277012129211\\
24	0.907630230740594\\
25	0.909038131119644\\
26	0.917277753261871\\
27	0.924184714797214\\
28	0.92978989503635\\
29	0.939950974558556\\
30	0.943608039969966\\
31	0.947823381557541\\
32	0.949822308995543\\
33	0.955630333983898\\
34	0.958129025734728\\
35	0.961407575196141\\
36	0.967013396378981\\
37	0.968678596845557\\
38	0.970726198735726\\
40	0.975612508228089\\
};
\addlegendentry{$u(t)\simiid \N(0, \frac{\gamma_u^2}{n_u}I_{n_u})$}

\addplot[area legend, draw=none, fill=mycolor3, fill opacity=0.1, forget plot]
table[row sep=crcr] {%
x	y\\
0	0.5\\
1	0.506743865898455\\
2	0.537144925023077\\
3	0.548070389884664\\
4	0.559242216674627\\
5	0.569532880647514\\
6	0.605553785151199\\
7	0.630786109227683\\
8	0.645037754918005\\
9	0.659772195173399\\
10	0.670756296223172\\
11	0.676817512433364\\
12	0.691075189275647\\
13	0.686049698425497\\
14	0.68376983819891\\
15	0.686266394596603\\
16	0.699284640540104\\
17	0.707008049954072\\
18	0.720091934067189\\
19	0.730044210275211\\
20	0.750844061672081\\
21	0.771166838148734\\
22	0.786877490681551\\
23	0.800045159938941\\
24	0.799795989253549\\
25	0.798410854417569\\
26	0.815915721698724\\
27	0.833231705336866\\
28	0.852289511354896\\
29	0.869459136543674\\
30	0.87700237585941\\
31	0.88736253489882\\
32	0.890517895789184\\
33	0.901603654240121\\
34	0.906796334798116\\
35	0.912559420397492\\
36	0.927516070849445\\
37	0.930511286170072\\
38	0.932445829301508\\
39	0.936902737201555\\
40	0.944591245377178\\
40	1.00663377107901\\
39	1.00943708991843\\
38	1.00900656816994\\
37	1.00684590752105\\
36	1.00651072190851\\
35	1.01025572999479\\
34	1.00946171667134\\
33	1.00965701372767\\
32	1.0091267222019\\
31	1.00828422821625\\
30	1.01021370408052\\
29	1.01044281257344\\
28	1.00729027871781\\
27	1.01513772425755\\
26	1.01863978482502\\
25	1.01966540782172\\
24	1.01546447222764\\
23	1.01050886431948\\
22	1.00677551824688\\
21	1.00703118636248\\
20	1.00265300710023\\
19	1.00627541288542\\
18	0.998904880314033\\
17	0.986811505193748\\
16	0.975777596573183\\
15	0.96071821441827\\
14	0.94940011544723\\
13	0.932036968605769\\
12	0.914221210626053\\
11	0.888840331083814\\
10	0.876679763204737\\
9	0.853938545133009\\
8	0.840271112216046\\
7	0.802929938170703\\
6	0.778442739028999\\
5	0.745246016235441\\
4	0.706036573587192\\
3	0.665348514028231\\
2	0.621442897609793\\
1	0.574867097660363\\
0	0.5\\
}--cycle;
\end{axis}

\end{tikzpicture}%
}
    \caption{Mean and $\sigma$-band of the likelihood of $\theta_*$ given the data for different data collection strategies for the motivating example.}
    \label{fig:numericsToyExample}
\end{figure}

\end{document}